\documentclass[12pt,reqno]{amsart}
\usepackage{fullpage,verbatim,paralist}
\usepackage[breaklinks,pdfstartview=FitH]{hyperref}
\newcommand{\Clust}{\mathrm{\bf Clust}}
\newcommand{\Alg}{\mathrm{Alg}}
\newcommand{\SDP}{\mathrm{SDP}}
\usepackage{amssymb,fullpage,mathrsfs,verbatim,graphicx,paralist,graphics}
\newcommand{\sign}{\mathrm{sign}}

\newcommand{\eqdef}{\stackrel{\mathrm{def}}{=}}
\renewcommand{\le}{\leqslant}

\renewcommand{\ge}{\geqslant}

\renewcommand{\geq}{\geqslant}
\renewcommand{\setminus}{\smallsetminus}

\newcommand{\N}{\mathbb{N}}
\newcommand{\E}{\mathbb{E}}
\newcommand{\Q}{\mathbb{Q}}

\newcommand{\R}{\mathbb{R}}

\newcommand{\T}{\mathcal{T}}

\newcommand{\Z}{\mathbb Z}

\newcommand{\e}{\varepsilon}

\theoremstyle{plain}
  \newtheorem{lemma}{Lemma}[section]

  \newtheorem{theorem}[lemma]{Theorem}

  \theoremstyle{definition}

  \newtheorem{remark}[lemma]{Remark}
  
\begin{document}

\title{Grothendieck-type inequalities\\ in combinatorial
optimization}\thanks{S.~K. was partially supported by NSF CAREER
grant CCF-0833228, NSF Expeditions grant CCF-0832795, an NSF
Waterman award, and BSF grant 2008059. A.~N. was partially supported
by NSF Expeditions grant CCF-0832795, BSF grant 2006009, and the
Packard Foundation.}

\author{Subhash Khot}
\address{Courant Institute, New York University, 251 Mercer Street, New York NY 10012, USA} \email{khot@cims.nyu.edu}

\author{Assaf Naor}
\address{Courant Institute, New York University, 251 Mercer Street, New York NY 10012, USA}
\email{naor@cims.nyu.edu}

\maketitle

\begin{abstract}
We survey connections of the Grothendieck inequality and its variants  to combinatorial optimization and computational complexity.
\end{abstract}

\tableofcontents

\section{Introduction}

The Grothendieck inequality asserts that there exists a universal
constant $K\in (0,\infty)$ such that for every $m,n\in \N$ and every
$m\times n$ matrix $A=(a_{ij})$ with real entries we have
\begin{multline}\label{eq:def gro}
\max\left\{\sum_{i=1}^m\sum_{j=1}^n a_{ij} \langle x_i,y_j\rangle:\ \{x_i\}_{i=1}^m,\{y_j\}_{j=1}^n\subseteq S^{n+m-1}\right\}\\ \le K \max\left\{\sum_{i=1}^m\sum_{j=1}^n a_{ij} \e_i\delta_j: \{\e_i\}_{i=1}^m,\{\delta_j\}_{j=1}^n\subseteq \{-1,1\}\right\}.
\end{multline}
Here, and in what follows, the standard scalar product on $\R^k$ is
denoted $\langle x,y\rangle =\sum_{i=1}^k x_i y_i$ and the Euclidean
sphere in $\R^k$ is denoted $S^{k-1} =\{x\in \R^k:\ \sum_{i=1}^k
x_i^2=1\}$. We refer to~\cite{DJT95,JL01} for the simplest known
proofs of the Grothendieck inequality; see
Section~\ref{sec:rounding} for a proof of~\eqref{eq:def gro}
yielding the best known bound on $K$. Grothendieck proved the
inequality~\eqref{eq:def gro} in~\cite{Gro53}, though it was stated
there in a different, but equivalent, form. The formulation of the
Grothendieck inequality appearing in~\eqref{eq:def gro} is due to
Lindenstrauss and Pe{\l}czy{\'n}ski~\cite{LP68}.

The Grothendieck inequality is of major importance to several areas,
ranging from Banach space theory to $C^*$ algebras and quantum
information theory. We will not attempt to indicate here this wide
range of applications of~\eqref{eq:def gro}, and refer instead
to~\cite{LP68,Tsi85,Pis86,Jam87,FR94,DJT95,Ble01,AGT06,Gar07,DFS08,Pit08,Pis11}
and the references therein. The purpose of this survey is to focus
solely on applications of the Grothendieck inequality and its
variants to combinatorial optimization, and to explain their
connections to computational complexity.

The infimum over those $K\in (0,\infty)$ for which~\eqref{eq:def gro} holds for all $m,n\in \N$ and all $m\times n$ matrices $A=(a_{ij})$ is called the Grothendieck constant, and is denoted $K_G$. Evaluating the exact value of $K_G$   remains a long-standing open problem, posed by Grothendieck in~\cite{Gro53}. In fact, even the second digit of $K_G$ is currently unknown, though clearly this is of lesser importance than the issue of understanding the structure of matrices $A$ and spherical configurations $\{x_i\}_{i=1}^m,\{y_j\}_{j=1}^n\subseteq S^{n+m-1}$ which make the inequality~\eqref{eq:def gro} ``most difficult".  Following a series of investigations~\cite{Gro53,LP68,Rie74,Krivine77,Kri79}, the best known upper bound~\cite{BMMN11} on $K_G$ is
\begin{equation}\label{eq:K_G current status is}
K_G<\frac{\pi}{2\log\left(1+\sqrt{2}\right)}=1.782...,
\end{equation}
and the best known lower bound~\cite{Ree91} on $K_G$ is
\begin{equation}\label{eq:reeds}
K_G\ge \frac{\pi}{2}e^{\eta_0^2}=1.676...,
\end{equation}
where $\eta_0=0.25573...$ is the unique solution  of the equation
$$
1-2\sqrt{\frac{2}{\pi}}\int_0^\eta e^{-z^2/2}dz=\frac{2}{\pi}e^{-\eta^2}.
$$
In~\cite{RS09} the problem of estimating $K_G$ up to an additive
error of $\e\in (0,1)$ was reduced to an optimization over a compact
space, and by exhaustive search over an appropriate net it was shown
that there exists an algorithm that computes $K_G$ up to an additive
error of $\e\in (0,1)$ in time $\exp(\exp(O(1/\e^3)))$. It does not
seem likely that this approach can yield computer assisted proofs of
estimates such as~\eqref{eq:K_G current status is}
and~\eqref{eq:reeds}, though to the best of our knowledge this has
not been attempted.

In the above discussion we focused on the classical Grothendieck
inequality~\eqref{eq:def gro}. However, the literature contains
several variants and extensions of~\eqref{eq:def gro} that have been
introduced for various purposes and applications in the decades
following Grothendieck's original work. In this survey we describe
some of these variants, emphasizing relatively recent developments
that yielded Grothendieck-type inequalities that are a useful tool
in the design of polynomial time algorithms for computing
approximate solutions of computationally hard optimization problems.
In doing so, we omit some important topics, including applications
of the Grothendieck inequality to communication complexity and
quantum information theory. While these research directions can be
viewed as dealing with a type of optimization problem, they are of a
different nature than the applications described here, which belong
to classical optimization theory. Connections to communication
complexity have already been covered in the survey of Lee and
Shraibman~\cite{LeS07}; we refer in addition
to~\cite{LMSS07,LSS09,LS09-1,LS09} for more information on this
topic. An explanation of the relation of the Grothendieck inequality
to quantum mechanics is contained in Section~19 of Pisier's
survey~\cite{Pis11}, the pioneering work in this direction being
that of Tsirelson~\cite{Tsi85}. An investigation of these questions
from a computational complexity point of view was initiated
in~\cite{CHTW04}, where it was shown, for example, how to obtain a
polynomial time algorithm for computing the entangled value of an
XOR game based on Tsirelson's work. We hope that the developments
surrounding applications of the Grothendieck inequality in quantum
information theory will eventually be surveyed separately by experts
in this area. Interested readers are referred
to~\cite{Tsi85,FR94,CHTW04,AGT06,Hey06,PWPVJ08,Pit08,KKMTV08,BBT09,LSS09,LS09,RT09,Pis11}.
Perhaps the most influential variants of the Grothendieck inequality
are its noncommutative generalizations. The noncommutative versions
in~\cite{Pisier78,Haa85} were conjectured by Grothendieck
himself~\cite{Gro53}; additional extensions to operator spaces are
extensively discussed in Pisier's survey~\cite{Pis11}. We will not
describe these developments here, even though we believe that they
might have applications to optimization theory. Finally,
multi-linear extensions of the Grothendieck inequality have also
been investigated in the literature; see for
example~\cite{Var74,Ton78,Ble79,Smi88} and especially Blei's
book~\cite{Ble01}. We will not cover this research direction since
its relation to classical combinatorial optimization has not (yet?)
been established, though there are recent investigations of
multi-linear Grothendieck inequalities in the context of quantum
information theory~\cite{PWPVJ08,LSS09}.

Being a mainstay of functional analysis, the Grothendieck inequality
might attract to this survey readers who are not familiar with
approximation algorithms and computational complexity. We wish to
encourage such readers to persist beyond this introduction so that
they will be exposed to, and hopefully eventually contribute to, the
use of analytic tools in combinatorial optimization. For this reason
we include Sections~\ref{sec:complexity assumptions}, \ref{sec:SDP}
below; two very basic introductory sections intended to quickly
provide
 background on computational complexity and convex programming for non-experts.

\subsection{Assumptions from computational complexity}\label{sec:complexity assumptions}

At present there are few unconditional results on the limitations of
polynomial time computation. The standard practice in this field is
to frame an impossibility result in  computational complexity by
asserting that the polynomial time solvability of a certain
algorithmic task would contradict a benchmark hypothesis. We briefly
describe below two key hypotheses of this type.

A graph $G=(V,E)$ is $3$-colorable if there exists a partition
$\{C_1,C_2,C_3\}$ of $V$ such that for every $i\in \{1,2,3\}$ and
$u,v\in C_i$ we have $\{u,v\}\notin E$. The $P\neq NP$ hypothesis
asserts that there is no polynomial time algorithm that takes an
$n$-vertex graph as input and determines whether or not it is
$3$-colorable. We are doing an injustice to this important question
by stating it this way, since it has many far-reaching equivalent
formulations. We refer to~\cite{GJ79,Sip97,Coo06} for more
information, but for non-experts it suffices to keep the above
simple formulation in mind.

When we say that assuming $P\neq NP$ no polynomial time algorithm
can perform a certain task $\mathcal T$ (e.g., evaluating the
maximum of a certain function up to a predetermined error) we mean
that given an algorithm $ALG$ that performs the task $\mathcal T$ one can
design an algorithm $ALG'$ that determines whether or not any input
graph is $3$-colorable while making at most polynomially many calls
to the algorithm $ALG$, with at most polynomially many additional
Turing machine steps.
Thus, if $ALG$ were a polynomial time algorithm then the same would
be true for $ALG'$, contradicting the $P\neq NP$ hypothesis. Such
results are called hardness results. The message that non-experts
should keep in mind is that a hardness result is nothing more than
the design of a new algorithm for $3$-colorability, and if one
accepts the $P\neq NP$ hypothesis then it implies that there must
exist inputs on which $ALG$ takes super-polynomial time to
terminate.

The Unique Games Conjecture (UGC) asserts that for every $\e\in
(0,1)$ there exists a prime $p=p(\e)\in \N$ such that no polynomial
time algorithm can perform the following task. The input is a system
of $m$ linear equations in $n$ variables $x_1,\ldots,x_n$, each of
which has the form $x_i-x_j\equiv c_{ij}\mod p$ (thus the input is
 $S\subseteq \{1,\ldots,n\}\times \{1,\ldots,n\}$ and
$\{c_{ij}\}_{(i,j)\in S}\subseteq \N$). The algorithm must determine
whether there exists an assignment of an integer value to each
variable $x_i$ such that at least $(1-\e)m$ of the equations are
satisfied, or whether no assignment of such values can satisfy more
than $\e m$ of the equations. If neither of these possibilities
occur, then an arbitrary output is allowed.

As in the case of $P\neq NP$, saying that assuming the UGC no
polynomial time algorithm can perform a certain task $\mathcal T$ is
the same as designing a polynomial time algorithm that solves the
above linear equations problem while making at most polynomially
many calls to a ``black box" that can perform the task $\mathcal T$.
The UGC was introduced in~\cite{Khot02}, though the above
formulation of it, which is equivalent to the original one, is due
to~\cite{KKMO07}. The use of the UGC as a hardness hypothesis has
become popular over the past decade; we refer to the
survey~\cite{Kho10} for more information on this topic.

To simplify matters (while describing all the essential ideas),  we
allow polynomial time algorithms to be randomized. Most (if not all)
of the algorithms described here can be turned into deterministic
algorithms, and corresponding hardness results can be stated equally
well in the context randomized or deterministic algorithms.  We will
ignore these distinctions, even though they are important. Moreover,
it is widely believed that in our context these distinctions do not
exist, i.e., randomness does not add computational power to
polynomial time algorithms; see for example the discussion of the
$NP\not \subseteq BPP$ hypothesis in~\cite{AB09}.

\subsection{Convex and semidefinite programming}\label{sec:SDP}
An important paradigm of optimization theory is that one can
efficiently optimize linear functionals over compact convex sets
that have a ``membership oracle". A detailed exposition of this
statement is contained in~\cite{GLS93}, but for the sake of
completeness we now quote the precise formulation of the results
that will be used in this article.

Let $K\subseteq \R^n$ be a compact convex set. We are also given a
point $z\in \Q^n$ and two radii $r,R\in (0,\infty)\cap \Q$ such that
$B(z,r)\subseteq K\subseteq B(z,R)$, where $B(z,t)=\{x\in \R^n:\
\|x-z\|_2\le t\}$. In what follows, stating that an algorithm is
polynomial means that we allow the running time to grow at most
polynomially in the number of bits required to represent the data
$(z,r,R)$. Thus, if, say, $z=0$, $r=2^{-n}$ and $R=2^n$ then the
running time will be polynomial in the dimension $n$. Assume that
there exists an algorithm $ALG$ with the following properties. The
input of $ALG$ is a vector $y\in \Q^n$ and  $\e\in (0,1)\cap \Q$.
The running time of $ALG$ is polynomial in $n$ and the number of
bits required to represent the data $(\e,y)$. The output of $ALG$ is
the assertion that either the distance of $y$ from $K$ is at most
$\e$, or that the distance of $y$ from the complement of $K$ is at
most $\e$. Then there exists an algorithm $ALG'$ that takes as input
a vector $c=(c_1,\ldots,c_n)\in \Q^n$ and $\e\in (0,1)\cap \Q$ and
outputs a vector $y=(y_1,\ldots,y_n)\in \R^n$ that is at distance at
most $\e$ from $K$ and for every $x=(x_1,\ldots,x_n)\in K$ that is
at distance greater than $\e$ from the complement of $K$ we have
$\sum_{i=1}^n c_iy_i\ge \sum_{i=1}^n c_i x_i-\e$. The running time
of $ALG'$ is allowed to grow at most polynomially in $n$ and the
number of bits required to represent the data $(z,r,R,c,\e)$. This
important result is due to~\cite{JN76}; we refer to~\cite{GLS93} for
an excellent account of this theory.

The above statement is a key tool in optimization, as it yields a
polynomial time method to compute the maximum of linear functionals
on a given convex body with arbitrarily good precision. We note
 the following special case of this method, known as
semidefinite programming. Assume that $n=k^2$ and think of $\R^n$ as
the space of all $k\times k$ matrices. Assume that we are given a
compact convex set $K\subseteq \R^n$ that satisfies the above
assumptions, and that for a given $k\times k$ matrix $(c_{ij})$ we
wish to compute in polynomial time (up to a specified additive
error) the maximum of $\sum_{i=1}^k\sum_{j=1}^k c_{ij}x_{ij}$ over
the set of symmetric positive semidefinite matrices $(x_{ij})$ that
belong to $K$. This can indeed be done, since determining whether a
given symmetric matrix is (approximately) positive semidefinite is
an eignevalue computation and hence can be performed in polynomial
time. The use of semidefinite programming to design approximation
algorithms is by now a deep theory of fundamental importance to
several areas of theoretical computer science. The
Goemans-Williamson MAX-CUT algorithm~\cite{GW95} was a key
breakthrough in this context. It is safe to say that after the
discovery of this
 algorithm the field of approximation algorithms
was transformed, and many subsequent results, including those
presented in the present article, can be described as attempts to
mimic the success of the Goemans-Williamson approach in other
contexts.

\section{Applications of the classical Grothendieck
inequality}\label{sec:classical}

The classical Grothendieck inequality~\eqref{eq:def gro} has
applications to algorithmic questions of central interest. These
applications will be described here in some detail. In
Section~\ref{sec:cut} we discuss the cut norm estimation problem,
whose relation to the Grothendieck inequality was first noted
in~\cite{AN06}. This is a generic combinatorial optimization problem
that contains well-studied questions as subproblems. Examples of its
usefulness are presented in Sections~\ref{sec:zemeredi},
\ref{sec:FK}, \ref{sec:acyclic}, \ref{sec:linear}.
Section~\ref{sec:rounding} is devoted to the rounding problem,
including the (algorithmic) method behind the proof of the best
known upper bound on the Grothendieck constant.

\subsection{Cut norm estimation}\label{sec:cut}
Let $A=(a_{ij})$ be an $m\times n$ matrix with real entries.  The cut norm of $A$ is defined as follows
\begin{equation}\label{eq:def cut norm}
\|A\|_{cut}= \max_{\substack{S\subseteq \{1,\ldots,m\}\\T\subseteq \{1,\ldots,n\}}}\left|\sum_{\substack{i\in S\\j\in T}}a_{ij}\right|.
\end{equation}
We will now explain how the Grothendieck inequality can be used to obtain a polynomial time algorithm for the following problem. The input is an $m\times n$ matrix $A=(a_{ij})$ with real entries, and the goal of the algorithm is to output in polynomial time a number $\alpha$ that is guaranteed to satisfy
\begin{equation}\label{eq:value goal}
\|A\|_{cut}\le \alpha\le C\|A\|_{cut},
\end{equation}
where $C$ is a (hopefully not too large) universal constant. A closely related algorithmic goal is to output in polynomial time two subsets $S_0\subseteq \{1,\ldots,m\}$ and $T_0\subseteq \{1,\ldots,n\}$ satisfying
\begin{equation}\label{eq:rounding goal}
\left|\sum_{\substack{i\in S_0\\j\in T_0}}a_{ij}\right|\ge \frac{1}{C}\|A\|_{cut}.
\end{equation}

The link to the Grothendieck inequality is made via two simple transformations. Firstly, define an $(m+1)\times (n+1)$ matrix $B=(b_{ij})$ as follows.

\begin{equation}\label{eq:def B'}
B=\begin{pmatrix} a_{11} &  a_{12}& \dots& a_{1n}& -\sum_{k=1}^n a_{1k}\\
  a_{21} & a_{22}& \dots& a_{2n} & -\sum_{k=1}^n a_{2k}\\
  \vdots & \vdots & \ddots & \vdots & \vdots\\
            a_{m1} & a_{m2} & \dots& a_{mn} &-\sum_{k=1}^n a_{mk}\\
              -\sum_{\ell=1}^m a_{\ell 1} & -\sum_{\ell=1}^m a_{\ell 2}& \dots &-\sum_{\ell=1}^m a_{\ell n}&\sum_{k=1}^n\sum_{\ell=1}^m a_{\ell k}
                       \end{pmatrix}.
\end{equation}
Observe that
\begin{equation}\label{eq:cut not changed}
\|A\|_{cut}=\|B\|_{cut}.
\end{equation}
 Indeed, for every $S\subseteq \{1,\ldots,m+1\}$ and $T\subseteq \{1,\ldots,n+1\}$ define $S^*\subseteq \{1,\ldots,m\}$ and $T^*\subseteq \{1,\ldots,n\}$ by
$$
S^*=\left\{\begin{array}{ll}S &\mathrm{if}\ m+1\notin S,\\
\{1,\ldots,m\}\setminus S& \mathrm{if}\ m+1\in S,\end{array}\right.\quad\mathrm{and}\quad T^*=\left\{\begin{array}{ll}T &\mathrm{if}\ n+1\notin T,\\
\{1,\ldots,n\}\setminus T& \mathrm{if}\ n+1\in T.\end{array}\right.
$$
One checks that for all $S\subseteq \{1,\ldots,m+1\}$ and $T\subseteq \{1,\ldots,n+1\}$ we have
$$
\left|\sum_{\substack{i\in S\\j\in T}}b_{ij}\right|=\left|\sum_{\substack{i\in S^*\\j\in T^*}}a_{ij}\right|,
$$
implying~\eqref{eq:cut not changed}. We next claim that
\begin{equation}\label{eq:operator norm}
\|B\|_{cut}=\frac14\|B\|_{\infty\to 1},
\end{equation}
where
\begin{equation}\label{eq:def infty to 1}
\|B\|_{\infty\to 1}=\max\left\{\sum_{i=1}^{m+1}\sum_{j=1}^{n+1} b_{ij} \e_i\delta_j: \{\e_i\}_{i=1}^{m+1},\{\delta_j\}_{j=1}^{n+1}\subseteq \{-1,1\}\right\}.
\end{equation}
To explain this notation observe that $\|B\|_{\infty\to 1}$ is the
norm of $B$ when viewed as a linear operator from $\ell_\infty^n$ to
$\ell_1^m$. Here, and in what follows, for $p\in [1,\infty]$ and
$k\in \N$ the space $\ell_p^k$ is $\R^k$ equipped with the $\ell_p$
norm $\|\cdot\|_p$, where $\|x\|_p^p=\sum_{\ell=1}^k |x_\ell|^p$ for
$x=(x_1,\ldots,x_k)\in \R^k$ (for $p=\infty$ we set as usual
$\|x\|_\infty=\max_{i\in \{1,\ldots,n\}}|x_i|$). Though it is
important, this operator theoretic interpretation of the quantity
$\|B\|_{\infty\to 1}$ will not have any role in this survey, so it
may be harmlessly ignored at first reading.

The proof of~\eqref{eq:operator norm} is simple: for $\{\e_i\}_{i=1}^{m+1},\{\delta_j\}_{j=1}^{n+1}\subseteq \{-1,1\}$ define $S^+,S^-\subseteq \{1,\ldots,m+1\}$ and $T^+,T^-\subseteq \{1,\ldots, n+1\}$ by setting $S^\pm=\{i\in \{1,\ldots,m+1\}:\ \e_i=\pm1\}$ and $T^\pm=\{j\in \{1,\ldots,n+1\}:\ \delta_j=\pm1\}$. Then
\begin{equation}\label{eq:4 bound}
\sum_{i=1}^{m+1}\sum_{j=1}^{n+1} b_{ij} \e_i\delta_j=\sum_{\substack{i\in S^+\\j\in T^+}} b_{ij}+\sum_{\substack{i\in S^-\\j\in T^-}} b_{ij}-\sum_{\substack{i\in S^+\\j\in T^-}} b_{ij}-\sum_{\substack{i\in S^-\\j\in T^+}} b_{ij}\le 4\|B\|_{cut}.
\end{equation}
This shows that $\|B\|_{\infty\to 1}\le 4\|B\|_{cut}$ (for any
matrix $B$, actually, not just the specific choice in~\eqref{eq:def
B'}; we will use this observation later, in
Section~\ref{sec:acyclic}). In the reverse direction, given
$S\subseteq \{1,\ldots,m+1\}$ and $T\subseteq \{1,\ldots,n+1\}$
define for $i\in \{1,\ldots,m+1\}$ and $j\in \{1,\ldots,n+1\}$,
$$
\e_i=\left\{\begin{array}{ll} 1&\mathrm{if}\ i\in S,\\ -1 &\mathrm{if}\ i\notin S,\end{array}\right.\quad\mathrm{and}\quad
\delta_j=\left\{\begin{array}{ll} 1&\mathrm{if}\ j\in T,\\ -1 &\mathrm{if}\ j\notin T.\end{array}\right.
$$
Then, since the sum of each row and each column of $B$ vanishes,
$$
\sum_{\substack{i\in S\\ j\in T}}b_{ij}=\sum_{i=1}^{m+1}\sum_{j=1}^{n+1} b_{ij}\frac{1+\e_i}{2}\cdot\frac{1+\delta_j}{2}=\frac14 \sum_{i=1}^{m+1}\sum_{j=1}^{n+1} b_{ij}\e_i\delta_j\le \frac14\|B\|_{\infty\to 1}.
$$
This completes the proof of~\eqref{eq:operator norm}. We summarize the above simple transformations in the following lemma.

\begin{lemma}\label{lem:new matrix}
Let $A=(a_{ij})$ be an $m\times n$ matrix with real entries and let $B=(b_{ij})$ be the $(m+1)\times (n+1)$ matrix given in~\eqref{eq:def B'}. Then
$$
\|A\|_{cut}=\frac14 \|B\|_{\infty\to 1}.
$$
\end{lemma}
A consequence of Lemma~\ref{lem:new matrix} is that the problem of approximating $\|A\|_{cut}$ in polynomial time is equivalent to the problem of approximating $\|A\|_{\infty \to 1}$ in polynomial time in the sense that any algorithm for one of these problems can be used to obtain an algorithm for the other problem with the same running time (up to constant factors) and the same (multiplicative) approximation guarantee.

Given an $m\times n$ matrix $A=(a_{ij})$ consider the following quantity.
\begin{equation}\label{eq:def sdp}
\SDP(A)=\max\left\{\sum_{i=1}^m\sum_{j=1}^n a_{ij} \langle x_i,y_j\rangle:\ \{x_i\}_{i=1}^m,\{y_j\}_{j=1}^n\subseteq S^{n+m-1}\right\}.
\end{equation}
The maximization problem in~\eqref{eq:def sdp} falls into the
framework of semidefinite programming as discussed in
Section~\ref{sec:SDP}. Therefore $\SDP(A)$ can be computed in
polynomial time with arbitrarily good precision. It is clear that
$\SDP(A)\ge \|A\|_{\infty\to 1}$, because the maximum
in~\eqref{eq:def sdp} is over a bigger set than the maximum
in~\eqref{eq:def infty to 1}. The Grothendieck inequality says that
$\SDP(A)\le K_G\|A\|_{\infty\to 1}$, so we have
$$
\|A\|_{\infty\to 1}\le \SDP(A)\le K_G\|A\|_{\infty\to 1}.
$$
Thus, the polynomial time algorithm that outputs the number $\SDP(A)$ is guaranteed to be within a factor of $K_G$ of $\|A\|_{\infty\to 1}$. By Lemma~\ref{lem:new matrix}, the algorithm that outputs the number $\alpha=\frac14\SDP(B)$, where the matrix $B$ is as in~\eqref{eq:def B'}, satisfies~\eqref{eq:value goal} with $C=K_G$.

Section~\ref{sec:hardness} is devoted to algorithmic impossibility results. But, it is worthwhile to make at this juncture two comments  regarding hardness of approximation. First of all, unless $P=NP$, we need to introduce an error $C>1$ in our requirement~\eqref{eq:value goal}. This was observed in~\cite{AN06}: the classical MAXCUT problem from algorithmic graph theory was shown in~\cite{AN06} to be a special case of the problem of computing $\|A\|_{cut}$, and therefore by~\cite{Haa01} we know that unless $P=NP$ there does not exist a polynomial time algorithm that outputs a number $\alpha$ satisfying~\eqref{eq:value goal} with $C$ strictly smaller than $\frac{17}{16}$. In fact, by a reduction to the MAX DICUT problem one can show that $C$ must be at least $\frac{13}{12}$, unless $P=NP$; we refer to Section~\ref{sec:hardness} and~\cite{AN06} for more information on this topic.

Another (more striking) algorithmic impossibility result is based on the Unique Games Conjecture (UGC). Clearly the above algorithm cannot yield an approximation guarantee strictly smaller than $K_G$ (this is the definition of $K_G$). In fact, it was shown in~\cite{RS09} that unless the UGC is false, for every $\e\in (0,1)$ any polynomial time algorithm for estimating $\|A\|_{cut}$ whatsoever, and not only the specific algorithm described above, must make an error of at least $K_G-\e$ on some input matrix $A$. Thus, if we assume the UGC then the classical Grothendieck constant has a complexity theoretic interpretation: it equals the best approximation ratio of polynomial time algorithms for the cut norm problem. Note that~\cite{RS09} manages to prove this statement  despite the fact that the value of $K_G$ is unknown.

We have thus far ignored the issue of finding in polynomial time the
subsets $S_0,T_0$ satisfying~\eqref{eq:rounding goal}, i.e., we only
explained how the  Grothendieck inequality can be used for
polynomial time estimation of the quantity $\|A\|_{cut}$ without
actually finding efficiently subsets at which $\|A\|_{cut}$ is
approximately attained. In order to do this we cannot use the
Grothendieck inequality as a black box: we need to look into its
proof and argue that it yields a polynomial time procedure that
converts vectors $\{x_i\}_{i=1}^m,\{y_j\}_{j=1}^n\subseteq
S^{n+m-1}$ into signs
$\{\e_i\}_{i=1}^m,\{\delta_j\}_{j=1}^n\subseteq \{-1,1\}$ (this is
known as a rounding procedure). It is indeed possible to do so, as
explained in Section~\ref{sec:rounding}. We postpone the explanation
of the rounding procedure that hides behind the Grothendieck
inequality in order to first give examples why one might want to
efficiently compute the cut norm of a matrix.

\subsubsection{Szemer\'edi partitions}\label{sec:zemeredi}
The Szemer\'edi regularity lemma~\cite{Sze78} (see also~\cite{KR03})
is a general and very useful structure theorem for graphs, asserting
(informally) that any graph can be partitioned into a controlled
number of pieces that interact with each other in a pseudo-random
way. The Grothendieck inequality, via the cut norm estimation
algorithm, yields a polynomial time algorithm that, when given a
graph $G=(V,E)$ as input, outputs a partition of $V$ that satisfies
the conclusion of the Szemer\'edi regularity lemma.

To make the above statements formal, we need to recall some definitions. Let $G=(V,E)$ be a graph. For every disjoint $X,Y\subseteq V$ denote the number of edges joining $X$ and $Y$ by $e(X,Y)=|\{(u,v)\in X\times Y:\ \{u,v\}\in E\}|$. Let $X,Y\subseteq V$ be disjoint and nonempty, and fix $\e,\delta\in (0,1)$. The pair of vertex sets $(X,Y)$ is called $(\e,\delta)$-regular if for every $S\subseteq X$ and $T\subseteq Y$ that are not too small, the quantity $\frac{e(S,T)}{|S|\cdot |T|}$ (the density of edges between $S$ and $T$) is essentially independent of the pair $(S,T)$ itself. Formally, we require that for every $S\subseteq X$ with $|S|\ge \delta |X|$ and every $T\subseteq Y$ with $|T|\ge \delta |Y|$ we have
\begin{equation}\label{eq:def regularity}
\left|\frac{e(S,T)}{|S|\cdot |T|}-\frac{e(X,Y)}{|X|\cdot |Y|}\right|\le \e.
\end{equation}
The almost uniformity of the numbers  $\frac{e(S,T)}{|S|\cdot |T|}$
as exhibited in~\eqref{eq:def regularity} says that the pair $(X,Y)$
is ``pseudo-random", i.e., it is similar to a random bipartite graph
where each $(x,y)\in X\times Y$ is joined by an edge independently
with probability $\frac{e(X,Y)}{|X|\cdot |Y|}$.

The Szemer\'edi regularity lemma says that for all $\e,\delta,\eta\in (0,1)$ and $k\in \N$ there exists $K=K(\e,\delta,\eta,k)\in \N$ such that for all $n\in \N$ any $n$-vertex graph $G=(V,E)$ can be partitioned into $m$-sets $S_1,\ldots,S_m\subseteq V$ with the following properties
\begin{itemize}
\item $k\le m\le K$,
\item $|S_i|-|S_j|\le 1$ for all $i,j\in \{1,\ldots,m\}$,
\item the number of $i,j\in \{1,\ldots, m\}$ with $i<j$ such that the pair $(S_i,S_j)$ is $(\e,\delta)$-regular is at least $(1-\eta)\binom{m}{2}$.
\end{itemize}
Thus  every graph is almost a superposition of a bounded number of pseudo-random graphs, the key point being that $K$ is independent of $n$ and the specific combinatorial structure of the graph in question.

It would be of interest to have a way to produce a Szemer\'edi
partition in polynomial time with $K$ independent of $n$ (this is a
good example of an approximation algorithm: one might care to find
such a partition into the minimum possible number of pieces, but
producing any partition into boundedly many pieces is already a
significant achievement). Such a polynomial time algorithm was
designed in~\cite{ADLRY94} (see also~\cite{KRT03}). We refer
to~\cite{ADLRY94,KRT03} for applications of algorithms for
constructing Szemer\'edi partitions, and to~\cite{ADLRY94} for a
discussion of the computational complexity of this algorithmic task.
We shall now explain how the Grothendieck inequality yields a
different approach to this problem, which has some advantages
over~\cite{ADLRY94,KRT03} that will be described later. The argument
below is due to~\cite{AN06}.

Assume that $X,Y$ are disjoint $n$-point subsets of a graph $G=(V,E)$. How can we determine in polynomial time whether or not the pair $(X,Y)$ is close to being $(\e,\delta)$-regular? It turns out that this is the main ``bottleneck" towards our goal to construct Szemer\'edi partitions in polynomial time. To this end consider the following $n\times n$ matrix $A=(a_{xy})_{(x,y)\in X\times Y}$.
\begin{equation}\label{eq:def A szem}
a_{xy}=\left\{\begin{array}{ll}1-\frac{e(X,Y)}{|X|\cdot |Y|}&\mathrm{if}\ \{x,y\}\in E,\\
-\frac{e(X,Y)}{|X|\cdot |Y|}& \mathrm{if}\ \{x,y\}\notin E.\end{array}\right.
\end{equation}
By the definition of $A$, if $S\subseteq X$ and $T\subseteq Y$ then
\begin{equation}\label{eq:szem idendity}
\left|\sum_{\substack{x\in S\\y\in T}} a_{xy}\right|=|S|\cdot|T|\cdot \left|\frac{e(S,T)}{|S|\cdot |T|}-\frac{e(X,Y)}{|X|\cdot |Y|}\right|.
\end{equation}
Hence if $(X,Y)$ is not $(\e,\delta)$-regular then $\|A\|_{cut}\ge \e\delta^2n^2$. The approximate cut norm  algorithm based on the Grothendieck inequality, together with the rounding procedure in Section~\ref{sec:rounding}, finds in polynomial time  subsets $S\subseteq X$ and $T\subseteq Y$ such that
\begin{equation*}
 \min \left\{n|S|,n|T|,n^2\left|\frac{e(S,T)}{|S|\cdot |T|}-\frac{e(X,Y)}{|X|\cdot |Y|}\right|\right\}\stackrel{\eqref{eq:szem idendity}}{\ge}\left|\sum_{\substack{x\in S\\y\in T}} a_{xy}\right|\ge \frac{1}{K_G}\e\delta^2n^2\ge \frac12 \e\delta^2n^2.
\end{equation*}
This establishes the following lemma.
\begin{lemma}\label{lem:alg reg}
There exists a polynomial time algorithm that takes as input two disjoint $n$-point subsets $X,Y$ of a graph, and either decides that $(X,Y)$ is $(\e,\delta)$-regular or finds $S\subseteq X$ and $T\subseteq Y$ with
$$
|S|,|T|\ge \frac12 \e\delta^2 n\quad \mathrm{and}\quad  \left|\frac{e(S,T)}{|S|\cdot |T|}-\frac{e(X,Y)}{|X|\cdot |Y|}\right|\ge \frac12 \e\delta^2.
$$
\end{lemma}
From Lemma~\ref{lem:alg reg} it is quite simple to design a polynomial algorithm that constructs a Szemer\'edi partition with bounded cardinality; compare Lemma~\ref{lem:alg reg} to Corollary 3.3 in~\cite{ADLRY94} and Theorem 1.5 in~\cite{KRT03}. We will not explain this deduction here since it is identical to the argument in~\cite{ADLRY94}. We note that the quantitative bounds in Lemma~\ref{lem:alg reg} improve over the corresponding bounds in~\cite{ADLRY94,KRT03} yielding, say, when $\e=\delta=\eta$, an algorithm with the best known bound on $K$ as a function of $\e$ (this bound is nevertheless still huge, as must be the case due to~\cite{Gow97}; see also~\cite{CF11}). See~\cite{AN06} for a precise statement of these bounds. In addition, the algorithms of~\cite{ADLRY94,KRT03} worked only in the ``dense case", i.e., when $\|A\|_{cut}$, for $A$ as in~\eqref{eq:def A szem}, is of order $n^2$, while the above algorithm does not have this requirement. This observation can be used to design the only known polynomial time algorithm for sparse versions of the Szemer\'edi regularity lemma~\cite{ACHKRS10} (see also~\cite{GS07}). We will not discuss the sparse version of the regularity lemma here, and refer instead to~\cite{Koh97,KR03} for a discussion of this topic. We also refer to~\cite{ACHKRS10} for additional applications of the Grothendieck inequality in sparse settings.

\subsubsection{Frieze-Kannan matrix decomposition}\label{sec:FK} The cut norm estimation problem was originally raised in the work of Frieze and Kannan~\cite{FK99} which introduced a method to design polynomial time approximation schemes for dense constraint satisfaction problems. The key tool for this purpose is a decomposition theorem for matrices that we now describe.

An $m\times n$ matrix $D=(d_{ij})$ is called a cut matrix if there exist subsets $S\subseteq \{1,\ldots,m\}$ and $T\subseteq \{1,\ldots,n\}$, and $d\in \R$ such that for all $(i,j)\in \{1,\ldots,m\}\times \{1,\ldots,n\}$ we have,
\begin{equation}\label{eq:def cut matrix}
d_{ij}=\left\{\begin{array}{ll}d &\mathrm{if}\ (i,j)\in S\times T,\\
0 &\mathrm{if}\ (i,j)\notin S\times T.\end{array}\right.
\end{equation}
Denote the matrix $D$ defined in~\eqref{eq:def cut matrix} by
$CUT(S,T,d)$. In~\cite{FK99} it is proved that for every $\e>0$
there exists an integer $s=O(1/\e^2)$ such that  for any $m\times n$
matrix $A=(a_{ij})$ with entries bounded in absolute value by $1$,
there are cut matrices $D_1,\ldots,D_s$ satisfying
\begin{equation}\label{eq:FK thm}
\left\|A-\sum_{k=1}^s D_k\right\|_{cut}\le \e mn.
\end{equation}
Moreover, these cut matrices $D_1,\ldots,D_s$ can be found in time $C(\e)(mn)^{O(1)}$. We shall now explain how this is done using the cut norm approximation algorithm of Section~\ref{sec:cut}.

The argument is iterative. Set $A_0=A$, and assuming that the cut matrices $D_1,\ldots,D_r$ have already been defined write $A_r=(a_{ij}(r))=A- \sum_{k=1}^r D_k$. We are done if $\left\|A_r\right\|_{cut}\le \e mn$, so we may assume that $\left\|A_r\right\|_{cut}> \e mn$. By the cut norm approximation algorithm we can find in polynomial time $S\subseteq \{1,\ldots, m\}$ and $T\subseteq \{1,\ldots, n\}$ satisfying
\begin{equation}\label{eq:use cut for FK}
\left|\sum_{\substack{i\in S\\j\in T}}a_{ij}(r)\right|\ge c \|A_r\|_{cut}\ge c\e mn,
\end{equation}
where $c>0$ is a universal constant. Set
$$
d=\frac{1}{|S|\cdot |T|} \sum_{\substack{i\in S\\j\in T}}a_{ij}(r).
$$
Define $D_{r+1}=CUT(S,T,d)$ and $A_{r+1}=(a_{ij}(r+1))=A_r-D_{r+1}$. Then by expanding the squares we have,
$$
\sum_{i=1}^m\sum_{j=1}^n a_{ij}(r+1)^2=\sum_{i=1}^m\sum_{j=1}^n a_{ij}(r)^2-\frac{1}{|S|\cdot |T|} \left(\sum_{\substack{i\in S\\j\in T}}a_{ij}(r)\right)^2\stackrel{\eqref{eq:use cut for FK}}{\le} \sum_{i=1}^m\sum_{j=1}^n a_{ij}(r)^2-c^2\e^2 mn.
$$
It follows inductively that if we can carry out this procedure $r$ times then
$$0\le\sum_{i=1}^m\sum_{j=1}^n a_{ij}(r)^2\le \sum_{i=1}^m\sum_{j=1}^n a_{ij}^2-rc^2\e^2mn\le mn-rc^2\e^2mn,
$$
where we used the assumption that $|a_{ij}|\le 1$. Therefore the
above iteration must terminate after $\lceil 1/(c^2\e^2)\rceil$
steps, yielding~\eqref{eq:FK thm}. We note that the bound
$s=O(1/\e^2)$ in~\eqref{eq:FK thm} cannot be improved~\cite{AFKK03};
see also~\cite{LS07,CF11} for related lower bounds.

The key step in the above algorithm was finding sets $S,T$ as in~\eqref{eq:use cut for FK}. In~\cite{FK99} an algorithm was designed that, given an $m\times n$ matrix $A=(a_{ij})$ and $\e>0$ as input, produces in time $2^{1/\e^{O(1)}}(mn)^{O(1)}$ subsets $S\subseteq \{1,\ldots,m\}$ and $T\subseteq \{1,\ldots,n\}$ satisfying
\begin{equation}\label{eq:additive}
\left|\sum_{\substack{i\in S\\j\in T}}a_{ij}\right|\ge \|A\|_{cut}-\e mn.
\end{equation}
The additive approximation guarantee in~\eqref{eq:additive}
implies~\eqref{eq:use cut for FK} only if $\|A\|_{cut}\ge
\e(c+1)mn$, and similarly the running time is not polynomial if,
say, $\e=n^{-\Omega(1)}$. Thus the Kannan-Frieze method is relevant
only to ``dense" instances, while the cut norm algorithm based on
the Grothendieck inequality applies equally well for all values of
$\|A\|_{cut}$. This fact, combined with more work (and, necessarily,
additional assumptions on the matrix $A$), was used in~\cite{CCF10}
to obtain a sparse version of~\eqref{eq:FK thm}: with $\e mn$ in the
right hand side of~\eqref{eq:FK thm} replaced by $\e\|A\|_{cut}$ and
$s=O(1/\e^2)$ (importantly, here $s$ is independent of $m,n$).

We have indicated above how the cut norm approximation problem is relevant to Kannan-Frieze matrix decompositions, but we did not indicate the uses of such decompositions since this is beyond the scope of the current survey. We refer to~\cite{FK99,AFKK03,BW09,CCF10} for a variety of applications of this methodology to combinatorial optimization problems.

\subsubsection{Maximum acyclic subgraph}\label{sec:acyclic} In the maximum acyclic subgraph problem we are given as input an $n$-vertex directed graph $G=(\{1,\ldots,n\},E)$.
Thus $E$ consists of a family of {\em ordered} pairs of distinct
elements in $\{1,\ldots,n\}$. We are interested in the maximum of
\begin{equation*}\label{eq:def acyclic}
\big|\{(i,j)\in \{1,\ldots,n\}^2:\ \sigma(i)<\sigma(j)\}\cap E\big|-
\big|\{(i,j)\in \{1,\ldots,n\}^2:\ \sigma(i)>\sigma(j)\}\cap E\big|
\end{equation*}
over all possible permutations $\sigma \in S_n$ ($S_n$ denotes the
group of permutations of $\{1,\ldots,n\}$). In words, the quantity
of interest is the maximum over all orderings of the vertices  of
the number of edges going ``forward" minus the number of edges going
``backward". The best known approximation algorithm for this problem
was discovered in~\cite{CMM07} as an application of the cut norm
approximation algorithm.

It is most natural to explain the algorithm of~\cite{CMM07} for a weighted version of the maximum acyclic subgraph problem. Let $W:\{1,\ldots,n\}\times \{1,\ldots,n\}\to \R$ be skew symmetric, i.e., $W(u,v)=-W(v,u)$ for all $u,v\in \{1,\ldots,n\}$.  For $\sigma\in S_n$ define
$$
W(\sigma)=\sum_{\substack{u,v\in \{1,\ldots,n\}\\u<v}}W(\sigma(u),\sigma(v)).
$$
Thus $W(\sigma)$ is the sum of the entries of $W$ that lie above the
diagonal after the rows and columns of $W$ have been permuted
according to the permutation $\sigma$. We are interested in the
quantity $ M_W=\max_{\sigma \in S_n} W(\sigma). $ The case of a
directed graph $G=(\{1,\ldots,n\},E)$ described above corresponds to
the matrix $W(u,v)=\mathbf{1}_{\{(u,v)\in
E\}}-\mathbf{1}_{\{(v,u)\in E\}}$.

\begin{theorem}[\cite{CMM07}]\label{thm:CMM}
The exists a polynomial time algorithm that takes as input an
$n\times n$ skew symmetric $W:\{1,\ldots,n\}\times \{1,\ldots,n\}\to
\R$ and outputs a permutation $\sigma\in S_n$
satisfying\footnote{Here, and in what follows, the relations
$\gtrsim,\lesssim$ indicate the corresponding inequalities up to an
absolute factor. The
relation $\asymp$ stands for $\gtrsim\wedge\lesssim$.} 
$$
W(\sigma)\gtrsim \frac{M_W}{\log n}.
$$
\end{theorem}

\begin{proof} The proof below is a slight variant of the reasoning of~\cite{CMM07}. By the cut norm approximation algorithm one can find in polynomial time two  subsets $S,T\subseteq \{1,\ldots,n\}$ satisfying
\begin{equation}\label{eq:no abs}
\sum_{\substack{u\in S\\v\in T}} W(u,v)\ge c \|W\|_{cut},
\end{equation}
where $c\in (0,\infty)$ is a universal constant. Note that we do not
need to take the absolute value of the left hand side
of~\eqref{eq:no abs} because $W$ is skew symmetric. Observe also
that since $W$ is skew symmetric we have $\sum_{u,v\in S\cap T}
W(u,v)=0$ and therefore
$$
\sum_{\substack{u\in S\\v\in T}} W(u,v)=\sum_{\substack{u\in
S\setminus T\\v\in T\setminus S}} W(u,v)+\sum_{\substack{u\in
S\setminus T \\v\in S\cap T}} W(u,v)+\sum_{\substack{u\in S\cap
T\\v\in T\setminus S}} W(u,v).
$$
By replacing the pair of subsets $(S,T)$ by one of $\{(S\setminus
T,T\setminus S),(S\setminus T,S\cap T),(S\cap T,T\setminus S)\}$,
and replacing the constant $c$ is~\eqref{eq:no abs} by $c/3$, we may
 assume without loss of generality that~\eqref{eq:no abs} holds
with $S$ and $T$ disjoint. Denote $R=\{1,\ldots,n\}\setminus (S\cup
T)$ and write $S=\{s_1,\ldots,s_{|S|}\}$, $T=\{t_1,\ldots,t_{|T|}\}$
and $R=\{r_1,\ldots,r_{|R|}\}$, where $s_1<\cdots<s_{|S|}$,
$t_1<\cdots<t_{|T|}$ and $r_1<\cdots<r_{|R|}$.

Define two permutations $\sigma^1,\sigma^2 \in S_n$ as follows.
$$
\sigma^1(u)=\left\{\begin{array}{ll}s_u&\mathrm{if}\ u\in \{1,\ldots,|S|\},\\
t_{u-|S|}&\mathrm{if}\ u\in \{|S|+1,\ldots,|S|+|T|\},\\
r_{u-|S|-|T|}&\mathrm{if}\ u\in
\{|S|+|T|+1,\ldots,n\},\end{array}\right.
$$
and
$$
\sigma^2(u)=\left\{\begin{array}{ll}r_{|R|-u+1}&\mathrm{if}\ u\in \{1,\ldots,|R|\},\\
s_{|R|+|S|-u+1}&\mathrm{if}\ u\in \{|R|+1,\ldots,|R|+|S|\},\\
t_{n-u+1}&\mathrm{if}\ u\in \{|R|+|S|+1,\ldots,n\}.\end{array}\right.
$$
In words, $\sigma^1$ orders $\{1,\ldots,n\}$ by starting with the
elements of $S$ in increasing order, then the elements of $T$ in
increasing order, and finally the elements of $R$ in increasing
order. At the same time, $\sigma^2$ orders $\{1,\ldots,n\}$ by
starting with the elements of $R$ in decreasing order, then the
elements of $S$ in decreasing order, and finally the elements of $T$
in decreasing order. The quantity $W(\sigma^1)+W(\sigma^2)$ consists
of a sum of terms of the form $W(u,v)$ for $u,v\in \{1,\ldots,n\}$,
where if $(u,v)\in (S\times S)\cup (T\times T)\cup (R\times
\{1,\ldots,n\})$ then both $W(u,v)$ and $W(v,u)$ appear exactly once
in this sum, and if $(u,v)\in S\times T$ then $W(u,v)$ appears twice
in this sum and $W(v,u)$ does not appear in this sum at all.
Therefore, using the fact that $W$ is skew symmetric we have the
following identity.
$$
W(\sigma^1)+W(\sigma^2)=2\sum_{\substack{u\in S\\v\in T}} W(u,v).
$$
It follows that for some $\ell\in \{1,2\}$ we have
$$
M(\sigma^\ell)\ge \sum_{\substack{u\in S\\v\in T}} W(u,v)\stackrel{\eqref{eq:no abs}}{\ge} c \|W\|_{cut}.
$$
The output of the algorithm will be the permutation $\sigma^\ell$, so it suffices to prove that
\begin{equation}\label{eq:CMM ineq}
\|W\|_{cut}\gtrsim  \frac{M_W}{\log n}.
\end{equation}
 We will prove
below that
\begin{equation}\label{eq:no permutation}
\|W\|_{cut}\gtrsim  \frac{1}{\log n}\sum_{\substack{u,v\in \{1,\ldots,n\}\\u<v}}W(u,v).
\end{equation}
Inequality~\eqref{eq:CMM ineq} follows by applying~\eqref{eq:no permutation} to  $W'(u,v)=W(\sigma(u),\sigma(v))$ for every $\sigma\in S_n$.

To prove~\eqref{eq:no permutation} first note that $\|W\|_{cut}\ge \frac14 \|W\|_{\infty\to 1}$; we have already proved this inequality as a consequence of the simple identity~\eqref{eq:4 bound}.   Moreover, we have
\begin{equation}\label{eq:imaginary}
\|W\|_{\infty\to 1}\gtrsim\max\left\{\sum_{u=1}^n\sum_{v=1}^n W(u,v) \sin(\alpha_u-\beta_v): \{\alpha_u\}_{u=1}^n,\{\beta_v\}_{v=1}^n\subseteq \R\right\}.
\end{equation}
Inequality~\eqref{eq:imaginary} is a special case of~\eqref{eq:def
gro} with the choice of vectors $x_u=(\sin \alpha_u,\cos
\alpha_u)\in \R^2$ and $y_v=(\cos\beta_v,-\sin\beta_v)\in \R^2$. We
note that this two-dimensional version of the Grothendieck
inequality is trivial with the constant in the right hand side
of~\eqref{eq:imaginary} being $\frac12$, and it is shown
in~\cite{Kri79} that the best constant in the right hand side
of~\eqref{eq:imaginary} is actually $\frac{1}{\sqrt{2}}$.



For every $\theta_1,\ldots,\theta_{n}\in \R$, an application of~\eqref{eq:imaginary} when $\alpha_u=\beta_u=\theta_u$ and $\alpha_u=\beta_u=-\theta_u$ yields the inequality
\begin{equation}\label{eq:sym}
\|W\|_{cut}\gtrsim \left|\sum_{u=1}^n\sum_{v=1}^n
W(u,v)\sin\left(\theta_u-\theta_v\right)\right|=
2\left|\sum_{\substack{u,v\in \{1,\ldots,n\}\\u<v}}
W(u,v)\sin\left(\theta_u-\theta_v\right)\right|,
\end{equation}
where for the equality in~\eqref{eq:sym} we used the fact that $W$
is skew symmetric. Consequently, for every $k\in\N$ we have
\begin{equation}\label{eq:use skew}
\|W\|_{cut}\gtrsim \left|\sum_{\substack{u,v\in
\{1,\ldots,n\}\\u<v}}W(u,v)\sin\left(\frac{\pi(v-u)k}{n}\right)\right|.
\end{equation}

By the standard orthogonality relation for the sine function, for every $u,v\in \{1,\ldots,n\}$ such that $u<v$ we have
\begin{equation}\label{eq:orthogonality}
\frac{2}{n}\sum_{k=1}^{n-1}\sum_{\ell=1}^{n-1} \sin\left(\frac{\pi(v-u)k}{n}\right)\sin\left(\frac{\pi k \ell}{n}\right)=1.
\end{equation}
Readers who are unfamiliar with~\eqref{eq:orthogonality} are
referred to its derivation in the appendix of~\cite{CMM07}; it can
be proved by substituting
$\sin\left(\pi(v-u)k/n\right)=(e^{i\pi(v-u)k/n}-e^{-i\pi(v-u)k/n})/(2i)$
and $\sin(\pi k\ell/n)=(e^{i\pi k\ell/n}-e^{-i\pi k\ell/n})/(2i)$
into the left hand side of~\eqref{eq:orthogonality} and computing
the resulting geometric sums explicitly.  Now,
\begin{eqnarray*}
\sum_{\substack{u,v\in \{1,\ldots,n\}\\u<v}}W(u,v)&\stackrel{\eqref{eq:orthogonality}}{=}&\frac{2}{n}\sum_{\substack{u,v\in \{1,\ldots,n\}\\u<v}}W(u,v)\sum_{k=1}^{n-1}\sum_{\ell=1}^{n-1} \sin\left(\frac{\pi(v-u)k}{n}\right)\sin\left(\frac{\pi k \ell}{n}\right)\\
&\le& \frac{2}{n}\sum_{k=1}^{n-1}\left|\sum_{\ell=1}^{n-1}\sin\left(\frac{\pi k \ell}{n}\right)\right|\cdot \left|\sum_{\substack{u,v\in \{1,\ldots,n\}\\u<v}}W(u,v)\sin\left(\frac{\pi(v-u)k}{n}\right)\right|\\&\stackrel{\eqref{eq:use skew}}{\lesssim}& \frac{\sum_{k=1}^{n-1}\left|\sum_{\ell=1}^{n-1}\sin\left(\frac{\pi k \ell}{n}\right)\right|}{n}\|W\|_{cut}.
\end{eqnarray*}
Hence, the desired inequality~\eqref{eq:no permutation} will follow
from $\sum_{k=1}^{n-1}\left|\sum_{\ell=1}^{n-1}\sin\left(\pi k
\ell/n\right)\right|\lesssim n\log n$. To establish this estimate
observe that by writing $\sin(\pi k\ell/n)=(e^{i\pi
k\ell/n}-e^{-i\pi k\ell/n})/(2i)$ and computing geometric sums
explicitly, one sees that $\sum_{\ell=1}^{n-1}\sin\left(\pi k
\ell/n\right)=0$ if $k$ is even and
$\sum_{\ell=1}^{n-1}\sin\left(\pi k \ell/n\right)=\cot(\pi k/(2n))$
if $k$ is odd (see the appendix of~\cite{CMM07} for the details of
this computation). Hence, since $\cot(\theta)< 1/\theta$ for every
$\theta\in (0,\pi/2)$, we have
\begin{equation*}
\sum_{k=1}^{n-1}\left|\sum_{\ell=1}^{n-1}\sin\left(\frac{\pi k \ell}{n}\right)\right|=\sum_{j=0}^{\left\lfloor \frac{n}{2}-1\right\rfloor}\cot\left(\frac{\pi (2j+1)}{2n}\right)\le \frac{2n}{\pi}\sum_{j=0}^{\left\lfloor \frac{n}{2}-1\right\rfloor} \frac{1}{2j+1}\lesssim n\log n\qedhere
\end{equation*}
\end{proof}

\subsubsection{Linear equations modulo 2}\label{sec:linear} Consider a system $\mathcal E$ of $N$ linear equations
modulo $2$ in $n$ Boolean variables $z_1,\ldots,z_n$ such that in
each equation appear only three distinct variables.
Let
$\mathrm{MAXSAT}(\mathcal E)$ be the maximum number of equations in
$\mathcal E$ that can be satisfied simultaneously. A random $\{0,1\}$
assignment of these variables satisfies in expectation $N/2$
equations, so it is natural to ask for
a polynomial time approximation algorithm to the quantity $\mathrm{MAXSAT}(\mathcal
E)-N/2$. We  describe below the best known~\cite{KN08} approximation algorithm for this problem, which uses the Grothendieck inequality in a crucial way. The approximation guarantee thus obtained is $O(\sqrt{n/\log n})$. While this allows for a large error, it is shown in~\cite{HV04} that for every $\e\in (0,1)$ if there were a polynomial time algorithm that approximates $\mathrm{MAXSAT}(\mathcal
E)-N/2$ to within a factor of $2^{(\log n)^{1-\e}}$ in time
$2^{(\log n)^{O(1)}}$ then there would be an algorithm for $3$-colorability that runs in time $2^{(\log n)^{O(1)}}$, a conclusion which is widely believed to be impossible.

Let $\mathcal E$ be a system of linear equations as described above. Write
$a_{ijk}=1$ if the equation $z_i+z_j+z_k=0$ is in the
system $\mathcal E$. Similarly write $a_{ijk}=-1$ if the
equation $z_i+z_j+z_k=1$ is in $\mathcal E$. Finally, write
$a_{ijk}=0$ if no equation in $\mathcal E$ corresponds
to $z_i+z_j+z_k$. Assume that the assignment $(z_1,\ldots,z_n)\in \{0,1\}^n$
satisfies $m$ of the equations in $\mathcal E$. Then
$$
\sum_{i=1}^n\sum_{j=1}^n\sum_{k=1}^n a_{ijk}
(-1)^{z_i+z_j+z_k}=m-(N-m)=2\left(m-\frac{N}{2}\right).
$$
 It follows
that
\begin{equation}\label{eq:pass to 3 tensor}
\max\left\{\sum_{i=1}^n\sum_{j=1}^n\sum_{k=1}^na_{ijk}\e_i\e_j\e_k:\ \{\e_i\}_{i=1}^n\subseteq \{-1,1\}\right\}=2\left(\mathrm{MAXSAT}(\mathcal
E)-\frac{N}{2}\right)\eqdef M.
\end{equation}

 We will now present a randomized polynomial algorithm that outputs a number $\alpha\in \R$ which satisfies with probability at least $\frac12$,
\begin{equation}\label{eq:new alpha goal}
\frac{1}{20K_G}\sqrt{\frac{\log n}{n}}M\le \alpha \le
M.
\end{equation}
Fix $m\in \N$ that will be determined later. Choose $\e^1,\ldots,\e^m\in \{-1,1\}^n$ independently and uniformly at random and consider the following random variable.
\begin{equation}\label{eq:random SDP}
\alpha=\frac{1}{10K_G}\max_{\ell\in \{1,\ldots,m\}}\max\left\{\sum_{i=1}^n\sum_{j=1}^n\sum_{k=1}^n a_{ijk}\e^\ell_i \langle y_j,z_k\rangle:\ \{y_j\}_{j=1}^n,\{z_k\}_{k=1}^n\subseteq S^{2n-1}\right\}.
\end{equation}
By the Grothendieck inequality we know that
\begin{equation}\label{eq:decouple}
\alpha\le \frac{1}{10}\max\left\{\sum_{i=1}^n\sum_{j=1}^n\sum_{k=1}^n a_{ijk} \e_i\delta_j\zeta_k:\ \{\e_i\}_{i=1}^n,\{\delta_j\}_{j=1}^n,\{\zeta_k\}_{k=1}^n\subseteq \{-1,1\}\right\}\le M.
\end{equation}
The final step in~\eqref{eq:decouple} follows from an elementary decoupling argument; see~\cite[Lem.~2.1]{KN08}.

 We claim that
\begin{equation}\label{eq:prob}
\Pr\left[\alpha \ge \frac{1}{20K_G} \sqrt{\frac{\log n}{n}}M \right]\ge 1-e^{-cm/\sqrt[4]{n}}.
\end{equation}
Once~\eqref{eq:prob} is established, it would follow that for $m\asymp \sqrt[4]{n}$ we have $\alpha\ge \frac{1}{20K_G}\sqrt{\frac{\log n}{n}}M$ with probability at least $\frac12$. This combined with~\eqref{eq:decouple} would complete the proof of~\eqref{eq:new alpha goal} since $\alpha$ as defined in~\eqref{eq:random SDP} can be computed in polynomial time, being the maximum of $O\left(\sqrt[4]{n}\right)$ semidefinite programs.

To check~\eqref{eq:prob} let $\|\cdot\|$ be the norm on $\R^n$ defined for every $x=(x_1,\ldots,x_n)\in \R^n$ by
$$
\|x\|=\max\left\{\sum_{i=1}^n\sum_{j=1}^n\sum_{k=1}^n a_{ijk}x_i \langle y_j,z_k\rangle:\ \{y_j\}_{j=1}^n,\{z_k\}_{k=1}^n\subseteq S^{2n-1}\right\}.
$$
Define $K=\{x\in \R^n:\ \|x\|\le 1\}$ and let $K^\circ=\{y\in \R^n:\ \sup_{x\in K}\langle x,y\rangle\le 1\}$ be the polar of $K$. Then $\max\{\|y\|_1:\ y\in K^\circ\}=\max\{\|x\|:\ \|x\|_\infty\le 1\}\ge M$, where the first equality is straightforward duality and the final inequality is a consequence of the definition of $\|\cdot\|$ and $M$. It follows that there exists $y\in K^\circ$ with  $\|y\|_1\ge M$. Hence,
\begin{multline*}
\Pr\left[\alpha \ge \frac{1}{20K_G}\sqrt{\frac{\log n}{n}}M \right]\stackrel{\eqref{eq:random SDP}}{=}1-\prod_{\ell=1}^m\Pr\left[\|\e^\ell\|<\frac12\sqrt{\frac{\log n}{n}}M\right]
\\\ge 1-\left(\Pr\left[\sum_{i=1}^n\e^1_iy_i<\frac12\sqrt{\frac{\log n}{n}}\sum_{i=1}^n|y_i|\right]\right)^m.
\end{multline*}
In order to prove~\eqref{eq:prob} it therefore suffices to prove that if $\e$ is chosen uniformly at random from $\{-1,1\}^n$ and $a\in \R^n$ satisfies $\|a\|_1=1$ then  $\Pr\left[\sum_{i=1}^n\e_ia_i\ge\sqrt{\log n/(4n)}\right]\ge 1-c/\sqrt[4]{n}$, where $c\in (0,\infty)$ is a universal constant. This probabilistic estimate for i.i.d. Bernoulli sums can be proved directly; see~\cite[Lem.~3.2]{KN08}.

\subsection{Rounding}\label{sec:rounding}

Let $A=(a_{ij})$ be an $m\times n$ matrix. In Section~\ref{sec:cut}
we described a polynomial time  algorithm for approximating
$\|A\|_{cut}$ and $\|A\|_{\infty\to 1}$. For applications it is also
important to find in polynomial time signs
$\e_1,\ldots,\e_m,\delta_1,\ldots,\delta_n\in \{-1,1\}$  for which
$\sum_{i=1}^m\sum_{j=1}^n a_{ij}\e_i\delta_j$ is at least a constant
multiple of $\|A\|_{\infty\to 1}$. This amounts to a ``rounding
problem": we need to find a procedure that, given vectors
$x_1,\ldots,x_m,y_1,\ldots,y_n\in S^{m+n-1}$, produces signs
$\e_1,\ldots,\e_m,\delta_1,\ldots,\delta_n\in \{-1,1\}$ whose
existence is ensured by the Grothendieck inequality, i.e.,
$\sum_{i=1}^m\sum_{j=1}^n a_{ij}\e_i\delta_j$ is at least a constant
multiple of $\sum_{i=1}^m\sum_{j=1}^n a_{ij}\langle x_i,y_j\rangle$.
For this purpose one needs to examine proofs of the Grothendieck
inequality, as done in~\cite{AN06}. We will now describe the
rounding procedure that gives the best known approximation
guarantee. This procedure yields a randomized algorithm that
produces the desired signs; it is also possible to obtain a
deterministic algorithm, as explained in~\cite{AN06}.

The argument below is based on a clever two-step rounding method due
to Krivine~\cite{Krivine77}. Fix $k\in \N$ and assume that we are
given two centrally symmetric measurable partitions of $\R^k$, or
equivalently two odd measurable functions $f,g:\R^k\to \{-1,1\}$.
Let $G_1,G_2\in \R^k$ be independent random vectors that are
distributed according to the standard Gaussian measure on $\R^k$,
i.e., the measure with density $x\mapsto
e^{-\|x\|_2^2/2}/(2\pi)^{k/2}$. For $t\in (-1,1)$ define
\begin{multline}\label{eq:h in intro}
H_{f,g}(t)\eqdef \E\left[f\left(\frac{1}{\sqrt{2}}G_1\right)g\left(\frac{t}{\sqrt{2}}G_1+\frac{\sqrt{1-t^2}}{\sqrt{2}}G_2\right)\right]\\
= \frac{1}{\pi^k(1-t^2)^{k/2}}\int_{\R^k}\int_{\R^k} f(x)g(y)\exp\left(\frac{-\|x\|_2^2-\|y\|_2^2+2t\langle x,y\rangle}{1-t^2}\right)dxdy.
\end{multline}
Then $H_{f,g}$ extends to an analytic function on the strip $\{z\in \mathbb C:\ \Re(z)\in (-1,1)\}$. The pair of functions  $\{f,g\}$ is called a Krivine rounding scheme if $H_{f,g}$ is invertible on a neighborhood of the origin, and if we consider the Taylor expansion $H_{f,g}^{-1}(z)=\sum_{j=0}^\infty a_{2j+1} z^{2j+1}$ then there exists $c=c(f,g)\in (0,\infty)$ satisfying $\sum_{j=0}^\infty |a_{2j+1}|c^{2j+1}=1$.

For $(f,g)$ as above and unit vectors $\{x_i\}_{i=1}^m,\{y_j\}_{j=1}^n\subseteq S^{m+n-1}$, one can find new unit vectors $\{u_i\}_{i=1}^m,\{v_j\}_{j=1}^n\subseteq S^{m+n-1}$ satisfying the identities
\begin{equation}\label{eq:use gen groth iden}
\forall (i,j)\in \{1,\ldots,m\}\times \{1,\ldots,n\},\quad \langle u_i,v_j\rangle = H_{f,g}^{-1}(c(f,g)\langle x_i,y_j\rangle).
\end{equation}
We refer to~\cite{BMMN11} for the proof that $\{u_i\}_{i=1}^m,\{v_j\}_{j=1}^n$ exist. This existence proof is not via an efficient algorithm, but as explained in~\cite{AN06}, once we know that they exist the new vectors can be computed efficiently provided $H_{f,g}^{-1}$ can be computed efficiently; this simply amounts to computing a Cholesky decomposition or, alternatively, solving a semidefinite program corresponding to~\eqref{eq:use gen groth iden}. This completes the first (preprocessing) step of a generalized Krivine rounding procedure. The next step is to apply a random projection to the new vectors thus obtained, as in Grothendieck's original proof~\cite{Gro53} or the Goemans-Williamson algorithm~\cite{GW95}.

Let $G: \R^{m+n}\to \R^k$ be a random $k\times (m+n)$ matrix whose entries are i.i.d.
standard Gaussian random variables. Define random signs $\{\e_i\}_{i=1}^m,\{\delta_j\}_{j=1}^n\subseteq \{-1,1\}$ by
\begin{equation}\label{eq:step 2}
\forall (i,j)\in \{1,\ldots,m\}\times \{1,\ldots,n\},\quad \e_i\eqdef f\left(\frac{1}{\sqrt{2}}Gu_i\right)\quad\mathrm{and}\quad \delta_j\eqdef g\left(\frac{1}{\sqrt{2}}Gv_j\right).
\end{equation}
Now,
\begin{equation}\label{eq:randomized rounding}
\E\left[\sum_{i=1}^m\sum_{j=1}^n a_{ij}\e_i\delta_j\right]\stackrel{(*)}{=} \E\left[\sum_{i=1}^m\sum_{j=1}^n a_{ij}
H_{f,g}\left(\langle u_i,v_j\rangle\right)\right]\stackrel{\eqref{eq:use gen groth iden}}{=}c(f,g)\sum_{i=1}^m\sum_{j=1}^n a_{ij}\langle x_i,y_j\rangle,
\end{equation}
where ($*$) follows by rotation invariance from~\eqref{eq:step 2}
and~\eqref{eq:h in intro}. The identity~\eqref{eq:randomized
rounding} yields the desired polynomial time randomized  rounding
algorithm, provided one can bound $c(f,g)$ from below. It also gives
a systematic way to bound the Grothendieck constant from above: for
every Krivine rounding scheme $f,g:\R^k\to \{-1,1\}$ we have $K_G\le
1/c(f,g)$. Krivine used this reasoning to obtain the bound $K_G\le
\pi/\left(2\log\left(1+\sqrt{2}\right)\right)$ by considering the
case $k=1$ and $f_0(x)=g_0(x)=\sign(x)$. One checks that
$\{f_0,g_0\}$ is a Krivine rounding scheme with
$H_{f_0,g_0}(t)=\frac{2}{\pi}\arcsin(t)$ (Grothendieck's identity)
and $c(f_0,g_0)=\frac{2}{\pi}\log\left(1+\sqrt{2}\right)$.

Since the goal of the above discussion is to round vectors
$\{x_i\}_{i=1}^m,\{y_j\}_{j=1}^n\subseteq  S^{m+n-1}$ to signs
$\{\e_i\}_{i=1}^m,\{\delta_j\}_{j=1}^n\subseteq \{-1,1\}$, it seems
natural to expect that the best possible Krivine rounding scheme
occurs when $k=1$ and $f(x)=g(x)=\sign(x)$. If true, this would
imply that $K_G= \pi/\left(2\log\left(1+\sqrt{2}\right)\right)$; a
long-standing conjecture of Krivine~\cite{Krivine77}. Over the years
additional evidence supporting Krivine's conjecture was discovered,
and a natural analytic conjecture was made in~\cite{Kon00} as a step
towards proving it. We will not discuss these topics here since
in~\cite{BMMN11} it was shown that actually $K_G\le
\pi/\left(2\log\left(1+\sqrt{2}\right)\right)-\e_0$ for some
effective constant $\e_0>0$.

It is known~\cite[Lem.~2.4]{BMMN11} that among all {\em one dimensional} Krivine rounding schemes $f,g:\R\to\{-1,1\}$ we indeed  have $c(f,g)\le \frac{2}{\pi}\log\left(1+\sqrt{2}\right)$, i.e., it does not pay off to take partitions of $\R$ which are more complicated than the half-line partitions. Somewhat unexpectedly, it was shown in~\cite{BMMN11} that a certain two dimensional Krivine rounding scheme $f,g:\R^2\to\{-1,1\}$ satisfies $c(f,g)>\frac{2}{\pi}\log\left(1+\sqrt{2}\right)$. The proof of~\cite{BMMN11} uses a Krivine rounding scheme $f,g:\R^2\to\{-1,1\}$ when $f=g$ corresponds to the partition of $\R^2$ as the sub-graph and super-graph of the polynomial $y=c\left(x^5-10x^3+15x\right)$, where $c>0$ is an appropriately chosen constant. This partition is depicted in Figure~\ref{fig:1}.

As explained in~\cite[Sec.~3]{BMMN11}, there is a natural guess for the ``best" two dimensional Krivine rounding scheme based on a certain numerical computation which we will not discuss here. For this (conjectural) scheme we have $f\neq g$, and the planar partition corresponding to $f$ is depicted in Figure~\ref{fig:2}.  Of course, once Krivine's conjecture has been disproved and the usefulness of higher dimensional rounding schemes has been established, there is no reason to expect that the situation won't improve as we consider $k$-dimensional Krivine rounding schemes for $k\ge 3$. A positive solution to an analytic question presented in~\cite{BMMN11} might even lead to an exact computation of $K_G$; see~\cite[Sec.~3]{BMMN11} for the details.

\begin{figure}[here]
\centering
\begin{tabular}{cc}
\begin{minipage}{217pt}
\frame{\includegraphics[width=217pt]{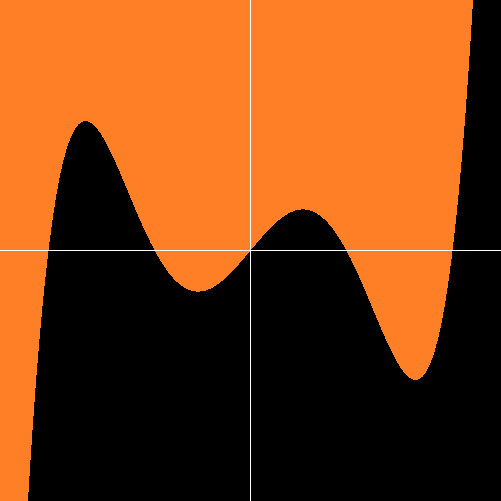}}
\caption{The partition of $\R^2$ used in~\cite{BMMN11} to show that $K_G$ is smaller than Krivine's bound;  the shaded regions are separated by the graph $y=c\left(x^5-10x^3+15x\right)$.}
\label{fig:1}
\end{minipage}
&
\begin{minipage}{217pt}
\frame{\includegraphics[width=217pt]{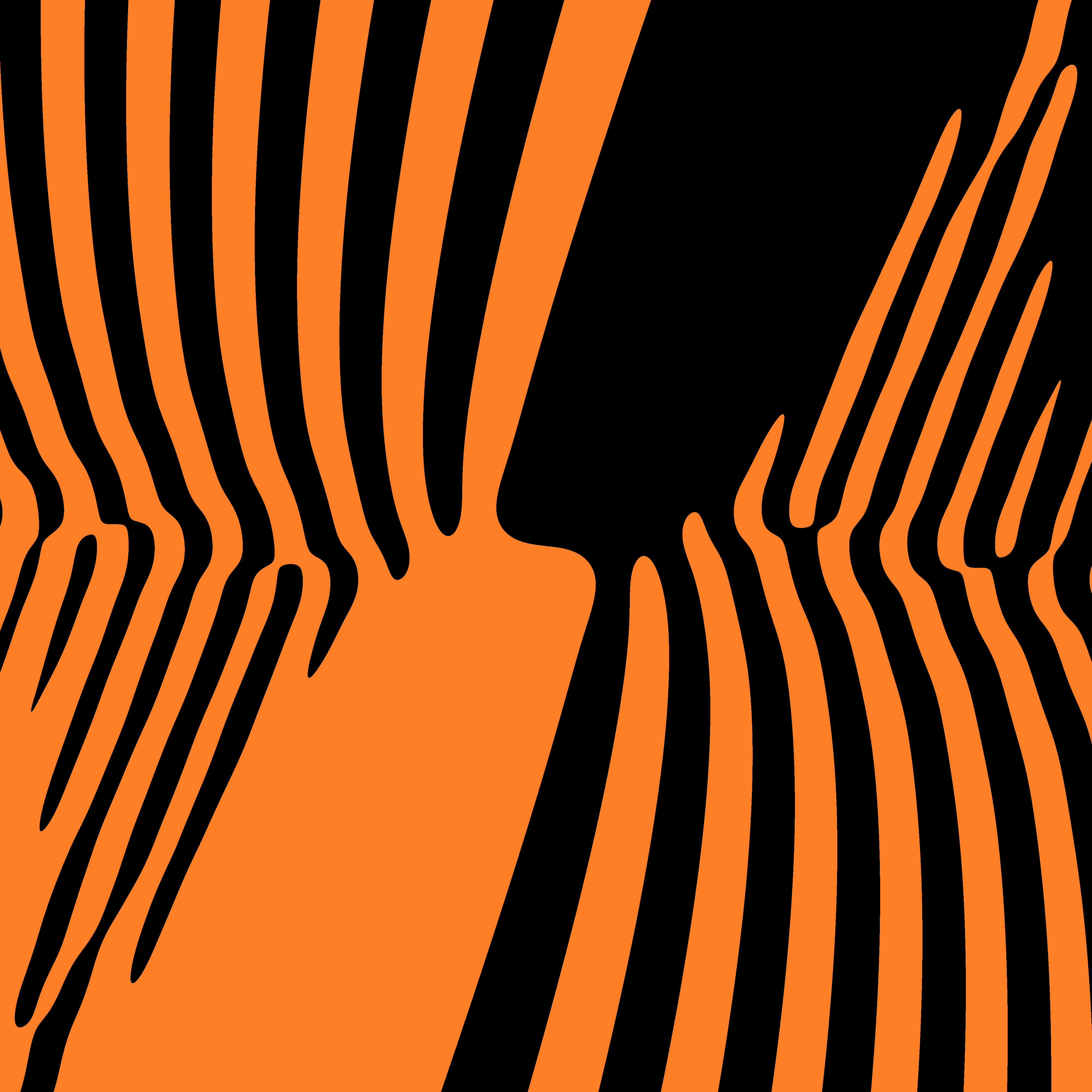}}
\caption{The ``tiger partition" restricted to the square
$[-20,20]^2$. This is the conjectured~\cite{BMMN11} optimal
partition of $\R^2$ for the purpose of Krivine-type rounding.}
\label{fig:2}
\end{minipage}
\end{tabular}
\end{figure}

\section{The Grothendieck constant of a graph}\label{sec:graphs}

Fix $n\in \N$ and let $G=(\{1,\ldots,n\},E)$ be a graph on the vertices $\{1,\ldots,n\}$. We assume throughout that $G$ does not contain any self loops, i.e., $E\subseteq \{S\subseteq \{1,\ldots,n\}:\ |S|=2\}$. Following~\cite{AMMN06}, define the Grothendieck constant of $G$, denoted $K(G)$, to be the smallest constant $K\in (0,\infty)$ such that every $n\times n$ matrix $(a_{ij})$ satisfies
\begin{equation}\label{eq:def gro graph}
\max_{x_1,\ldots,x_n\in S^{n-1}} \sum_{\substack{i,j\in \{1,\ldots,n\}\\\{i,j\}\in E}}a_{ij}\langle x_i,x_j\rangle\le K \max_{\e_1,\ldots,\e_n\in \{-1,1\}} \sum_{\substack{i,j\in \{1,\ldots,n\}\\\{i,j\}\in E}}a_{ij}\e_i\e_j.
\end{equation}
Inequality~\eqref{eq:def gro graph} is an extension of the
Grothendieck inequality since~\eqref{eq:def gro} is the special case
of~\eqref{eq:def gro graph} when $G$ is a bipartite graph. Thus
\begin{equation}\label{eq:bipartite}
K_G=\sup_{n\in \N}\left\{K(G):\ G\ \mathrm{is\ an\ }
n\mathrm{-vertex\ bipartite\ graph}\right\}.
\end{equation}

The opposite extreme of bipartite graphs is $G=K_n$, the $n$-vertex
complete graph. In this case~\eqref{eq:def gro graph} boils down to
the following inequality
\begin{equation}\label{eq:complete}
\max_{x_1,\ldots,x_n\in S^{n-1}} \sum_{\substack{i,j\in
\{1,\ldots,n\}\\i\neq j}}a_{ij}\langle x_i,x_j\rangle\le K(K_n)
 \max_{\e_1,\ldots,\e_n\in \{-1,1\}}
\sum_{\substack{i,j\in \{1,\ldots,n\}\\i\neq j}}a_{ij}\e_i\e_j.
\end{equation}
It turns out that $K(K_n)\asymp \log n$. The estimate
$K(K_n)\lesssim \log n$ was proved in~\cite{NRT99,Meg01,KS03,CW04}.
In fact, as shown in~\cite[Thm.~3.7]{AMMN06}, the following stronger
inequality holds true for every $n\times n$ matrix $(a_{ij})$; it
implies that $K(K_n)\lesssim \log n$ by the Cauchy-Schwartz inequality.
\begin{multline*}
\max_{x_1,\ldots,x_n\in S^{n-1}} \sum_{\substack{i,j\in
\{1,\ldots,n\}\\i\neq j}}a_{ij}\langle x_i,x_j\rangle\\\lesssim
\log\left(\frac{\sum_{i\in \{1,\ldots,n\}}\sum_{j\in
\{1,\ldots,n\}\setminus \{i\}}|a_{ij}|}{\sqrt{\sum_{i\in
\{1,\ldots,n\}}\sum_{j\in \{1,\ldots,n\}\setminus
\{i\}}a_{ij}^2}}\right) \max_{\e_1,\ldots,\e_n\in \{-1,1\}}
\sum_{\substack{i,j\in \{1,\ldots,n\}\\i\neq j}}a_{ij}\e_i\e_j.
\end{multline*}
The matching lower bound $K(K_n)\gtrsim \log n$ is due
to~\cite{AMMN06}, improving over a result of~\cite{KS03}.

How can we interpolate between the two extremes~\eqref{eq:bipartite}
and~\eqref{eq:complete}? The Grothendieck constant $K(G)$ depends on
the combinatorial structure of the graph $G$, but at present our
understanding of this dependence is incomplete. The following
general bounds are known.
\begin{equation}\label{eq:omega theta}
\log \omega\lesssim K(G)\lesssim \log \vartheta,
\end{equation}
and
\begin{equation}\label{eq:vallentin}
K(G)\le
\frac{\pi}{2\log\left(\frac{1+\sqrt{(\vartheta-1)^2+1}}{\vartheta-1}\right)},
\end{equation}
where \eqref{eq:omega theta} is due to~\cite{AMMN06}
and~\eqref{eq:vallentin} is due to~\cite{BOV10}.
 Here $\omega$ is the clique number of $G$, i.e., the largest
$k\in \{2,\ldots, n\}$ such that there exists $S\subseteq
\{1,\ldots,n\}$ of cardinality $k$ satisfying $\{i,j\}\in E$ for all
distinct $i,j\in S$, and
\begin{equation}\label{eq:def lovasz}
\vartheta=\min\left\{\max_{i\in \{1,\ldots,n\}}\frac{1}{\langle x_i,y\rangle^2}:\ x_1,\ldots,x_n,y\in S^{n}\ \wedge\ \forall\{i,j\}\in E,\  \langle x_i,x_j\rangle=0\right\}.
\end{equation}

The parameter $\vartheta$ is known as the Lov\'asz theta function of
the complement of $G$; an important graph parameter that was
introduced in~\cite{Lov79}. We refer to~\cite{KMS98}
and~\cite[Thm.~3.5]{AMMN06} for alternative characterizations of
$\vartheta$. It suffices to say here that it was shown
in~\cite{Lov79} that $\vartheta\le \chi$, where $\chi$ is the
chromatic number of $G$, i.e., the smallest integer $k$ such that
there exists a partition $\{A_1,\ldots,A_k\}$ of $\{1,\ldots,n\}$
such that $\{i,j\}\notin E$ for all $(i,j)\in \bigcup_{\ell=1}^k
A_\ell\times A_\ell$.  Note that the upper bound in~\eqref{eq:omega
theta} is superior to~\eqref{eq:vallentin} when $\vartheta$ is
large, but when $\vartheta=2$ the bound~\eqref{eq:vallentin} implies
Krivine's classical bound~\cite{Krivine77} $K_G\le
\pi/\left(2\log\left(1+\sqrt{2}\right)\right)$.

The upper and lower bounds in~\eqref{eq:omega theta} are known to
match up to absolute constants for a variety of graph classes.
Several such sharp Grothendieck-type inequalities are presented in
Sections 5.2 and 5.3 of~\cite{AMMN06} . For example, as explained
in~\cite{AMMN06}, it follows from~\eqref{eq:omega theta}, combined
with combinatorial results of~\cite{Lov79,AO95}, that for every
$n\times n\times n$ $3$-tensor $(a_{ijk})$ we have
$$
\max_{\{x_{ij}\}_{i,j=1}^n\subseteq
S^{n^2-1}}\sum_{\substack{i,j,k\in \{1,\ldots,n\}\\i\neq j\neq k}}
a_{ijk}\left\langle x_{ij},x_{jk}\right\rangle\lesssim
\max_{\{\e_{ij}\}_{i,j=1}^n\subseteq
\{-1,1\}}\sum_{\substack{i,j,k\in \{1,\ldots,n\}\\i\neq j\neq k}}
a_{ijk} \e_{ij}\e_{jk}.
$$

While~\eqref{eq:omega theta} is often a satisfactory asymptotic
evaluation of $K(G)$, this isn't always the case. In particular, it
is unknown whether $K(G)$ can be bounded from below by a function of
$\vartheta$ that tends to $\infty$ as $\vartheta\to \infty$. An
instance in which~\eqref{eq:omega theta} is not sharp is the case of
Erd\H{o}s-R\'enyi~\cite{ER60} random graphs $G(n,1/2)$. For such
graphs we have $\omega\asymp \log n$ almost surely as $n\to \infty$;
see~\cite{Mat70} and~\cite[Sec.~4.5]{AS00}. At the same time, for
$G(n,1/2)$ we have~\cite{Juh82} $\vartheta\asymp\sqrt{n}$ almost
surely as $n\to \infty$. Thus~\eqref{eq:omega theta} becomes in this
case the rather weak estimate $\log\log n\lesssim
K(G(n,1/2))\lesssim \log n$. It turns out~\cite{AB08} that
$K(G(n,1/2))\asymp \log n$ almost surely as $n\to \infty$; we refer
to~\cite{AB08} for additional computations of this type of the
Grothendieck constant of random and psuedo-random graphs. An
explicit evaluation of the Grothendieck constant of certain graph
families can be found in~\cite{LV11}; for example, if $G$ is a graph
of girth $g$ that is not a forest and does not admit $K_5$ as a
minor then $ K(G)=\frac{g\cos(\pi/g)}{g-2}$.

\subsection{Algorithmic consequences}\label{sec:graph alg} Other than being  a natural variant
of the Grothendieck inequality, and hence of intrinsic mathematical
interest, \eqref{eq:def gro graph} has ramifications to discrete
optimization problems, which we now describe.

\subsubsection{Spin glasses}\label{sec:spin}
Perhaps the most natural interpretation of~\eqref{eq:def gro graph}
is in the context of solid state physics, specifically the problem
of efficient computation of ground states of Ising spin glasses. The graph $G$ represents the
interaction pattern of $n$ particles; thus
$\{i,j\}\notin E$ if and only if the particles $i$ and $j$ cannot
interact with each other. Let $a_{ij}$ be the magnitude of the
interaction of $i$ and $j$ (the sign of $a_{ij}$ corresponds to
attraction/repulsion). In the Ising model each particle $i\in
\{1,\ldots,n\}$ has a spin $\e_i\in \{-1,1\}$ and the total energy
of the system is given by the quantity $-\sum_{\{i,j\}\in E}
a_{ij}\e_i\e_j$. A spin configuration $(\e_1,\ldots,\e_n)\in
\{-1,1\}^n$ is called a ground state if it minimizes the total
energy. Thus the problem of finding a ground state is precisely that
of computing the maximum appearing in the right hand side
of~\eqref{eq:def gro graph}. For more information on this topic
see~\cite[pp.~352--355]{LP86}.

Physical systems seek to settle at a ground state, and therefore it
is natural to ask whether it is computationally efficient (i.e.,
polynomial time computable) to find such a ground state, at least
approximately. Such questions have been studied in the physics
literature for several decades; see~\cite{BMRU80,Bar82,Bac84,BBT09}.
In particular, it was shown in~\cite{Bar82} that if $G$ is a planar
graph then one can find a ground state in polynomial time, but
in~\cite{Bac84} it was shown that when  $G$ is the three dimensional
grid then this computational task is NP-hard.

Since the quantity in the left hand side of~\eqref{eq:def gro graph}
is a semidefinite program and therefore can be computed in
polynomial time with arbitrarily good precision, a good bound on
$K(G)$ yields a polynomial time algorithm that computes the energy
of a ground state with correspondingly good approximation guarantee.
Moreover, as explained in~\cite{AMMN06}, the proof of the upper
bound in~\eqref{eq:omega theta} yields a polynomial time algorithm
that finds a spin configuration $(\sigma_1,\ldots,\sigma_n)\in
\{-1,1\}^n$ for which
\begin{equation}\label{eq:ground found}
\sum_{\substack{i,j\in \{1,\ldots,n\}\\\{i,j\}\in E}}
a_{ij}\sigma_i\sigma_j\gtrsim \frac{1}{\log
\vartheta}\cdot\max_{\{\e_i\}_{i=1}^n\subseteq
\{-1,1\}}\sum_{\substack{i,j\in \{1,\ldots,n\}\\\{i,j\}\in E}}
a_{ij}\e_i\e_j.
\end{equation}
An analogous polynomial time algorithm corresponds to the
bound~\eqref{eq:vallentin}. These algorithms yield the best known
efficient methods for computing a ground state of Ising spin glasses
on a variety of interaction graphs.

\subsubsection{Correlation clustering}\label{sec:correlation}
A different interpretation of~\eqref{eq:def gro graph} yields the
best known polynomial time approximation algorithm for the
correlation clustering problem~\cite{BBC02,CGV05}; this connection
is due to~\cite{CW04}. Interpret the graph $G=(\{1,\ldots,n\},E)$ as
the ``similarity/dissmilarity graph" for the items $\{1,\ldots,n\}$,
in the following sense. For $\{i,j\}\in E$ we are given a sign
$a_{ij}\in \{-1,1\}$ which has the following meaning: if $a_{ij}=1$
then $i$ and $j$ are deemed to be similar, and if $a_{ij}=-1$ then
$i$ and $j$ are deemed to be different. If $\{i,j\}\notin E$ then we
do not express any judgement on the similarity or dissimilarity of
$i$ and $j$.

Assume that $A_1,\ldots,A_k$ is a partition (or ``clustering") of
$\{1,\ldots,n\}$. An agreement between this clustering  and our
similarity/dissmilarity judgements is a pair $i,j\in \{1,\ldots,n\}$
such that $a_{ij}=1$ and $i,j\in A_r$ for some $r\in \{1,\ldots,k\}$
or $a_{ij}=-1$ and $i\in A_r$, $j\in A_s$ for distinct $r,s\in
\{1,\ldots,k\}$. A disagreement between this clustering  and our
similarity/dissmilarity judgements is a pair $i,j\in \{1,\ldots,n\}$
such that $a_{ij}=1$ and $i\in A_r$, $j\in A_s$ for distinct $r,s\in
\{1,\ldots,k\}$ or $a_{ij}=-1$ and $i,j\in A_r$ for some $r\in
\{1,\ldots,k\}$. Our goal is to cluster the items while encouraging
agreements and penalizing disagreements. Thus, we wish to find a
clustering of $\{1,\ldots,n\}$ into an unspecified number of
clusters which maximizes the total number of agreements minus the
total number of disagreements.

It was proved in~\cite{CW04} that the case of clustering into two
parts is the ``bottleneck" for this problem: if there were a
polynomial time algorithm that finds a clustering into two parts for
which the total number of agreements minus the total number of
disagreements is at least a fraction $\alpha\in (0,1)$ of the maximum  possible (over
all bi-partitions) total number of agreements minus the
total number of disagreements, then one could find in polynomial
time a clustering which is at least a fraction $\alpha/(2+\alpha)$
of the analogous maximum that is defined without specifying the number of clusters.

One checks that the problem of finding a partition into two clusters
that maximizes the total number of agreements minus the total number
of disagreements is the same as the problem of  computing the
maximum in the right hand side of~\eqref{eq:def gro graph}. Thus the
upper bound in~\eqref{eq:omega theta} yields a polynomial time
algorithm for correlation clustering with approximation guarantee
$O(\log \vartheta)$, which is the best known approximation algorithm
for this problem. Note that when $G$ is the complete graph then the
approximation ratio is $O(\log n)$. As will be explained in
Section~\ref{sec:hardness}, it is known~\cite{KS11} that for every
$\gamma\in (0,1/6)$, if there were a polynomial time algorithm for
correlation clustering that yields an approximation guarantee of
$(\log n)^\gamma$ then there would be an algorithm for
$3$-colorability that runs in time $2^{(\log n)^{O(1)}}$, a
conclusion which is widely believed to be impossible.

\section{Kernel clustering and the propeller conjecture}\label{sec:kernel}

Here we describe a large class of Grothendieck-type inequalities
that is motivated by algorithmic applications to a combinatorial
optimization problem called Kernel Clustering. This problem
originates in machine learning~\cite{SSGB07}, and its only known
rigorous approximation algorithms follow from Grothendieck
inequalities (these algorithms are sharp assuming the UGC). We will
first describe the inequalities and then the algorithmic
application.

Consider the special case of the Grothendieck
inequality~\eqref{eq:def gro} where $A = (a_{ij})$ is an $n\times n$
positive semidefinite matrix. In this case we may assume without
loss of generality that in~\eqref{eq:def gro} $x_i = y_i$ and $\e_i=
\delta_i$ for every $i\in \{1,\ldots,n\}$ since this holds for the
maxima on either side of~\eqref{eq:def gro} (see also the
explanation in~\cite[Sec.~5.2]{AN06}).  It follows
from~\cite{Gro53,Rie74} (see also~\cite{Nes98}) that
 for every $n\times n$ symmetric positive semidefinite
matrix $A=(a_{ij})$ we have
\begin{equation}\label{eq:PSD gro}
\max_{x_1,\ldots,x_n\in S^{n-1}}\sum_{i=1}^n\sum_{j=1}^n a_{ij}
\langle x_i,x_j\rangle \le \frac{\pi}{2}\cdot\max_{\e_1,\ldots,\e_n\in
\{-1,1\}}\sum_{i=1}^n\sum_{j=1}^n a_{ij}\e_i\e_j,
\end{equation}
and that $\frac{\pi}{2}$ is the best possible constant in~\eqref{eq:PSD gro}.

 A natural variant
of~\eqref{eq:PSD gro}  is to replace the numbers $-1,1$ by general
vectors $v_1,\ldots,v_k\in \R^k$, namely one might ask for the
smallest constant $K\in (0,\infty)$ such that for every symmetric
positive semidefinite $n\times n$ matrix $(a_{ij})$ we have:
\begin{equation}\label{eq:our gro with K intro}
\max_{x_1,\ldots,x_n\in S^{n-1}}\sum_{i=1}^n\sum_{j=1}^n a_{ij}
\langle x_i,x_j\rangle \le K\max_{u_1,\ldots,u_n\in
\{v_1,\ldots,v_k\}}\sum_{i=1}^n\sum_{j=1}^n a_{ij}\langle
u_i,u_j\rangle.
\end{equation}
The best constant $K$ in~\eqref{eq:our gro with K intro} can be
characterized as follows. Let $B=\left(b_{ij}= \langle
v_i,v_j\rangle\right)$ be the Gram matrix of $v_1,\ldots,v_k$. Let
$C(B)$ be the maximum over all partitions $\{A_1,\ldots,A_k\}$ of
$\R^{k-1}$ into measurable sets of the quantity
$\sum_{i=1}^k\sum_{j=1}^k b_{ij}\langle z_i,z_j\rangle$, where for
$i\in \{1,\ldots,k\}$ the vector $z_i\in \R^{k-1}$ is the Gaussian
moment of $A_i$, i.e.,
$$z_i=\frac{1}{(2\pi)^{(k-1)/2}}\int_{A_i}xe^{-\|x\|_2^2/2}dx.$$
It was proved in~\cite{KN10} that~\eqref{eq:our gro with K intro}
holds with $K=1/C(B)$ and that this constant is sharp.

Inequality~\eqref{eq:our gro with K intro} with $K=1/C(B)$ is proved
via the following {\it rounding procedure}. Fix unit vectors
$x_1,\ldots,x_n\in S^{n-1}$. Let $G=(g_{ij})$ be a $(k-1)\times n$
random matrix whose entries are i.i.d. standard Gaussian random
variables. Let $A_1,\ldots,A_k\ \subseteq \R^{k-1}$ be a measurable
partition of $\R^{k-1}$ at which $C(B)$ is attained (for a proof
that the maximum defining $C(B)$ is indeed attained,
see~\cite{KN10}). Define a random choice of $u_i\in
\{v_1,\ldots,v_k\}$ by setting $u_i=v_\ell$ for the unique $\ell\in
\{1,\ldots,k\}$ such that $Gx_i\in A_\ell$. The fact
that~\eqref{eq:our gro with K intro} holds with $K=1/C(B)$ is a
consequence of the following fact, whose proof we skip (the full
details are in~\cite{KN10}).
\begin{equation}\label{eq:in expectation}
\E\left[\sum_{i=1}^n\sum_{j=1}^n a_{ij}\langle
u_i,u_j\rangle\right]\ge C(B)\sum_{i=1}^n\sum_{j=1}^n a_{ij} \langle
x_i,x_j\rangle.
\end{equation}

Determining the partition of $\R^{k-1}$ that achieves the value
$C(B)$ is a nontrivial problem in general, even in the special case
when $B = I_k$ is the $k\times k$ identity matrix. Note that in this
case one desires a partition $\{A_1,\ldots,A_k\}$ of $\R^{k-1}$ into
measurable sets so as to maximize the following quantity.
 $$ \sum_{i=1}^k  \left\| \frac{1}{(2\pi)^{(k-1)/2}}\int_{A_i}xe^{-\|x\|_2^2/2}dx  \right\|_2^2.$$
As shown in \cite{KN09, KN10}, the optimal partition is given by
simplicial cones centered at the origin. When $B= I_2$ we have
$C(I_2) = \frac{1}{\pi}$, and the optimal partition of $\R$ into two
cones is the positive and the negative axes. When $B=I_3$ it was
shown in~\cite{KN09} that $C(I_3)=\frac{9}{8\pi}$, and the optimal
partition of $\R^2$ into three cones is the {\it propeller}
partition, i.e., into three cones with angular measure $120^\circ$
each.

Though it might be surprising at first sight, the authors posed
in~\cite{KN09} the {\it propeller conjecture}: for any $k \geq 4$,
the optimal partition of $\R^{k-1}$ into $k$ parts is ${\mathcal P}
\times \R^{k-3}$ where ${\mathcal P}$
 is the propeller partition of $\R^2$.  In other
words, even if one is allowed to use $k$ parts, the propeller
conjecture asserts that the best partition consists of only three
nonempty parts. Recently, this conjecture was solved
positively~\cite{HJN11} for $k=4$, i.e., for partitions of $\R^3$
into four measurable parts. The proof of~\cite{HJN11} reduces the problem to a concrete finite set of numerical inequalities which are then verified with full rigor in a computer-assisted fashion.  Note that this is the first nontrivial
(surprising?) case of the propeller conjecture, i.e., this is the
first case in which we indeed drop one of the four allowed parts in
the optimal partition.



We now describe an application of~\eqref{eq:our gro with K intro} to
the Kernel Clustering problem; a general framework for clustering
massive statistical data so as to uncover a certain hypothesized
structure~\cite{SSGB07}. The problem is defined as follows. Let
$A=(a_{ij})$ be an $n\times n$ symmetric positive semidefinite
matrix which is usually normalized to be centered, i.e.,
$\sum_{i=1}^n\sum_{j=1}^n a_{ij}=0$. The matrix $A$ is often thought
of as the correlation matrix of random variables $(X_1,\ldots,X_n)$
that measure attributes of certain empirical data, i.e.,
$a_{ij}=\E\left[X_iX_j\right]$. We are also given another symmetric
positive semidefinite $k\times k$ matrix $B=(b_{ij})$ which
functions as a hypothesis, or test matrix. Think of $n$ as huge and
$k$ as a small constant. The goal is to cluster $A$ so as to obtain
a smaller matrix which most resembles $B$. Formally, we wish to find
a partition $\{S_1,\ldots,S_k\}$ of $\{1,\ldots,n\}$ so that if we
write $c_{ij}=\sum_{(p,q)\in S_i\times S_j} a_{pq}$ then the
resulting clustered version of $A$ has the maximum correlation
$\sum_{i=1}^k\sum_{j=1}^k c_{ij}b_{ij}$ with the hypothesis matrix
$B$. In words, we form a $k\times k$ matrix $C=(c_{ij})$ by summing
the entries of $A$ over the blocks induced by the given partition,
and we wish to produce in this way a matrix that is most correlated
with $B$. Equivalently, the goal is to evaluate the number:
\begin{equation}\label{eq:def clust}
\Clust(A|B)=
\max_{\sigma:\{1,\ldots,n\}\to\{1,\ldots,k\}}\sum_{i=1}^k\sum_{j=1}^k
a_{ij}b_{\sigma(i)\sigma(j)}.
\end{equation}

The strength of this generic clustering framework is based in part
on the flexibility of adapting the matrix $B$ to the problem at
hand. Various particular choices of $B$ lead to well studied
optimization problems, while other specialized choices of $B$ are
based on statistical hypotheses which have been applied with some
empirical success. We refer to~\cite{SSGB07,KN09} for additional
background and a discussion of specific examples.


In~\cite{KN09} it was shown that there exists a randomized
polynomial time algorithm that takes as input two positive
semidefinite matrices $A,B$ and outputs a number $\alpha$ that
satisfies $\Clust(A|B)\le \E[\alpha]\le
\left(1+\frac{3\pi}{2}\right)\Clust(A|B)$. There is no reason to
believe that the approximation factor of $1+\frac{3\pi}{2}$ is
sharp, but nevertheless prior to this result, which is based
on~\eqref{eq:our gro with K intro}, no constant factor polynomial
time approximation algorithm for this problem was known.

Sharper results can be obtained if we assume that the input matrices
are normalized appropriately. Specifically, assume that $k\ge 3$ and
restrict only to inputs $A$ that are centered, i.e.,
$\sum_{i=1}^n\sum_{j=1}^n a_{ij}=0$, and inputs $B$ that are either
the identity matrix $I_k$, or satisfy
$\sum_{i=1}^k\sum_{j=1}^kb_{ij}=0$ ($B$ is centered as well) and
$b_{ii}=1$ for all $i\in \{1,\ldots,k\}$ ($B$ is ``spherical").
Under these assumptions the output of the algorithm of~\cite{KN09}
satisfies $\Clust(A|B)\le \E[\alpha]\le
\frac{8\pi}{9}\left(1-\frac{1}{k}\right)\Clust(A|B)$. Moreover, it
was shown in~\cite{KN09} that assuming the propeller conjecture and
the UGC, no polynomial time algorithm can achieve an approximation
guarantee that is strictly smaller than
$\frac{8\pi}{9}\left(1-\frac{1}{k}\right)$ (for input matrices
normalized as above). Since the propeller conjecture is known to
hold true for $k=3$ \cite{KN09} and $k=4$ \cite{HJN11}, we know that
the UGC hardness threshold for the above problem is exactly
$\frac{16\pi}{27}$ when $k=3$ and $\frac{2\pi}{3}$ when $k=4$.

A finer, and perhaps more natural, analysis of the kernel clustering
problem can be obtained if we fix the matrix $B$ and let the input
be only the matrix $A$, with the goal being, as before, to
approximate the quantity $\Clust(A|B)$ in polynomial time. Since $B$
is symmetric and positive semidefinite we can find vectors
$v_1,\ldots,v_k\in \R^k$ such that $B$ is their Gram matrix, i.e.,
$b_{ij}=\langle v_i,v_j\rangle$ for all $i,j\in \{1,\ldots,k\}$. Let
$R(B)$ be the smallest possible radius of a Euclidean ball in $\R^k$
which contains $\{v_1,\ldots,v_k\}$ and let $w(B)$ be the center of
this ball. We note that both  $R(B)$ and $w(B)$ can be efficiently
computed by solving an appropriate semidefinite program. Let $C(B)$
be the parameter defined above.


It is shown in~\cite{KN10} that for every fixed symmetric positive
semidefinite $k\times k$ matrix $B$ there exists a randomized
polynomial time algorithm which given an $n\times n$ symmetric
positive semidefinite centered matrix $A$, outputs a number
$\Alg(A)$ such that
$$
\Clust(A|B)\le \E\left[\Alg(A)\right]\le
\frac{R(B)^2}{C(B)}\Clust(A|B).
$$
As we will explain in Section~\ref{sec:hardness}, assuming the UGC
no polynomial time algorithm can achieve an approximation guaranty
strictly smaller than $R(B)^2/C(B)$.



The algorithm of~\cite{KN10}  uses semidefinite programming to
compute the value
\begin{multline}\label{def:SDP}
\SDP(A|B)= \max\left\{\sum_{i=1}^n\sum_{j=1}^n
a_{ij}\left\langle x_i,x_j\right \rangle :\ x_1,\ldots,x_n\in \R^n\
\wedge\ \|x_i\|_2\le 1\ \forall
i\in\{1,\ldots,n\}\right\}\\=\max\left\{\sum_{i=1}^n\sum_{j=1}^n
a_{ij}\left\langle x_i,x_j\right \rangle :\ x_1,\ldots,x_n\in
S^{n-1}\right\},
\end{multline}
where the last equality in~\eqref{def:SDP} holds since the function
$(x_1,\ldots,x_n)\mapsto \sum_{i=1}^n\sum_{j=1}^n a_{ij}\left\langle
x_i,x_j\right \rangle$ is convex (by virtue of the fact that $A$ is
positive semidefinite). We claim that
\begin{equation}\label{eq:guarantee}
\frac{\Clust(A|B)}{R(B)^2}\le \SDP(A|B)\le \frac{\Clust(A|B)}{C(B)},
\end{equation}
which implies that if we output the number $R(B)^2\SDP(A|B)$ we will
obtain a polynomial time algorithm which approximates $\Clust(A|B)$
up to a factor of $\frac{R(B)^2}{C(B)}$. To
verify~\eqref{eq:guarantee} let $x_1^*,\ldots,x_n^*\in S^{n-1}$ and
$\sigma^*:\{1,\ldots,n\}\to \{1,\ldots,k\}$ be such that
$$
\SDP(A|B)=\sum_{i=1}^n\sum_{j=1}^n a_{ij}\left\langle
x_i^*,x_j^*\right \rangle \quad \mathrm{and}\quad
\Clust(A|B)=\sum_{i=1}^n\sum_{j=1}^n a_{ij} b_{\sigma^*(i)
\sigma^*(j)}.
$$

Write $(a_{ij})_{i,j=1}^n=(\langle u_i,u_j\rangle)_{i,j=1}^n$ for
some $u_1,\ldots,u_n\in\R^n$. The assumption that $A$ is centered
means that $\sum_{i=1}^nu_i=0$. The rightmost  inequality
in~\eqref{eq:guarantee} is just the Grothendieck inequality
\eqref{eq:our gro with K intro}. The leftmost inequality
in~\eqref{eq:guarantee} follows from the fact that
$\frac{v_{\sigma^*(i)}-w(B)}{R(B)}$ has norm at most $1$ for all
$i\in \{1,\ldots,n\}$. Indeed, these norm bounds imply that
\begin{eqnarray*}
\SDP(A|B)&\ge&  \sum_{i=1}^n\sum_{j=1}^n a_{ij}\left\langle
\frac{v_{\sigma^*(i)}-w(B)}{R(B)},\frac{v_{\sigma^*(j)}-w(B)}{R(B)}\right
\rangle\\
&=& \frac{1}{R(B)^2} \sum_{i=1}^n\sum_{j=1}^n a_{ij}\left\langle
v_{\sigma^*(i)},v_{\sigma^*(j)}\right
\rangle\\&&-\frac{2}{R(B)^2}\sum_{i=1}^n\left\langle
w(B),v_{\sigma^*(i)}\right\rangle\left\langle u_i,\sum_{j=1}^nu_j
\right\rangle+\frac{\|w(B)\|_2^2}{R(B)^2} \sum_{i=1}^n\sum_{j=1}^n
a_{ij}\\
&=& \frac{\Clust(A|B)}{R(B)^2}.
\end{eqnarray*}

This completes the proof that the above algorithm approximates
efficiently the number $\Clust(A|B)$, but does not address the issue
of how to efficiently compute an assignment
$\sigma:\{1,\ldots,n\}\to \{1,\ldots,k\}$ for which the induced
clustering of $A$ has the required value.  The issue here is to find
efficiently a conical simplicial partition $A_1,\ldots,A_k$ of
$\R^{k-1}$ at which $C(B)$ is attained. Such a partition exists and
may be assumed to be hardwired into the description of the
algorithm. Alternately, the partition that achieves $C(B)$ up to a
desired degree of accuracy can be found by brute-force for fixed $k$
(or $k=k(n)$ growing sufficiently slowly as a function of $n$);
see~\cite{KN10}. For large values of $k$ the problem of computing
$C(B)$ efficiently remains open.


\section{The $L_p$ Grothendieck problem}\label{sec:L_p}

Fix $p\in [1,\infty]$ and consider the following algorithmic
problem. The input is an $n\times n$ matrix $A=(a_{ij})$ whose
diagonal entries vanish, and the goal is to compute (or estimate) in
polynomial time the quantity
\begin{equation}\label{eq:def: M_p}
M_p(A)=\max_{\substack{t_1,\ldots,t_n\in \R\\\sum_{k=1}^n|t_k|^p\le
1}}\sum_{i=1}^n\sum_{j=1}^n a_{ij} t_i
t_j=\max_{\substack{t_1,\ldots,t_n\in \R\\\sum_{k=1}^n|t_k|^p=
1}}\sum_{i=1}^n\sum_{j=1}^n a_{ij} t_i t_j.
\end{equation}
The second equality in~\eqref{eq:def: M_p} follows from a
straightforward convexity argument since the diagonal entries of $A$
vanish. Some of the results described below hold true without the
vanishing diagonal assumption, but we will tacitly make this
assumption here since the second equality in~\eqref{eq:def: M_p}
makes the problem become purely combinatorial when $p=\infty$.
Specifically, if $G=(\{1,\ldots,n\},E)$ is the complete graph then
$M_\infty(A)=\max_{\e_1,\ldots,\e_n\in \{-1,1\}}\sum_{\{i,j\}\in E}
a_{ij}\e_i\e_j$. The results described in Section~\ref{sec:graphs}
therefore imply that there is a polynomial time algorithm that
approximates $M_\infty(A)$ up to a $O(\log n)$ factor, and that it
is computationally hard to achieve an approximation guarantee
smaller than $(\log n)^\gamma$ for all $\gamma\in (0,1/6)$.

There are values of $p$ for which the above problem can be solved in
polynomial time. When $p=2$ the quantity $M_2(A)$ is the largest
eigenvalue of $A$, and hence can be computed in polynomial
time~\cite{GV96,LO96}. When $p=1$ it was shown in~\cite{Alo06} that
it is possible to approximate $M_1(A)$ up to a factor of $1+\e$ in
time $n^{O(1/\e)}$. It is also shown in~\cite{Alo06} that the
problem of $(1+\e)$-approximately  computing $M_1(A)$ is $W[1]$ complete; we refer
to~\cite{DF99} for the definition of this type of hardness result
and just say here that it indicates that a running time of
$c(\e)n^{O(1)}$ is impossible.

The algorithm of~\cite{Alo06} proceeds by showing that for every
$m\in \N$ there exist $y_1,\ldots,y_n\in \frac1{m}\Z$ with
$\sum_{i=1}^n|y_i|\le 1$ and $\sum_{i=1}^n\sum_{j=1}^n
a_{ij}y_iy_j\ge \left(1-\frac{1}{m}\right)M_1(A)$. The number of
such vectors $y$ is
$1+\sum_{k=1}^m\sum_{\ell=1}^k2^\ell\binom{n}{\ell}\binom{k-1}{\ell-1}\le
4n^m$. An exhaustive search over all such vectors will then
approximate $M_1(A)$ to within a factor of $m/(m-1)$ in time
$O(n^m)$. To prove the existence of $y$ fix $t_1,\ldots,t_n\in \R$
with $\sum_{k=1}^n |t_k|= 1$ and $\sum_{i=1}^n\sum_{j=1}^n a_{ij}
t_i t_j=M_1(A)$. Let $X\in \R^n$ be a random vector given by $
\Pr\left[X=\sign(t_j)e_j\right]=|t_j|$ for every $j\in
\{1,\ldots,n\}$. Here $e_1,\ldots,e_n$ is the standard basis of
$\R^n$. Let $\{X_s=(X_{s1},\ldots,X_{sn})\}_{s=1}^m$ be independent
copies of $X$ and set
$Y=(Y_1,\ldots,Y_n)=\frac{1}{m}\sum_{s=1}^mX_s$. Note that if
$s,t\in \{1,\ldots,m\}$ are distinct then for all $i,j\in
\{1,\ldots, n\}$ we have
$\E\left[X_{si}X_{tj}\right]=\sign(t_i)\sign(t_j)|t_i|\cdot|t_j|=t_it_j$.
Also, for every $s\in \{1,\ldots,m\}$ and every distinct $i,j\in
\{1,\ldots,n\}$ we have $X_{si}X_{sj}=0$. Since the diagonal entries
of $A$ vanish it follows that
\begin{equation}\label{eq:noga}
\E\left[\sum_{i=1}^n\sum_{j=1}^n a_{ij}
Y_iY_j\right]=\frac{1}{m^2}\sum_{\substack{s,t\in
\{1,\ldots,m\}\\s\neq t}} \sum_{\substack{i,j\in
\{1,\ldots,n\}\\i\neq j}}
a_{ij}\E\left[X_{si}X_{tj}\right]=\left(1-\frac{1}{m}\right)M_1(A).
\end{equation}
Noting that the vector $Y$ has $\ell_1$ norm at most $1$ and all of
its entries are integer multiples of $1/m$, it follows
from~\eqref{eq:noga} that with positive probability $Y$ will have
the desired properties.

How can we interpolate between the above results for $p\in
\{1,2,\infty\}$? It turns out that there is a satisfactory answer
for $p\in (2,\infty)$ but the range $p\in (1,2)$ remains a mystery.
To explain this write
$\gamma_p=\left(\E\left[|G|^p\right]\right)^{1/p}$, where $G$ is a
standard Gaussian random variable. One computes that
\begin{equation}\label{eq:def gamma_p}
\gamma_p=\sqrt{2}\left(\frac{\Gamma\left(\frac{p+1}{2}\right)}{\sqrt{\pi}}\right)^{1/p}.
\end{equation}
Also, Stirling's formula implies that $\gamma_p^2=\frac{p}{e}+O(1)$
as $p\to \infty$. It follows from~\cite{NS09,GRSW10} that for every
fixed $p\in [2,\infty)$ there exists a polynomial time algorithm
that approximates $M_p(A)$ to within a factor of $\gamma_p^2$, and
that for every $\e\in (0,1)$ the existence of a polynomial time
algorithm that approximates $M_p(A)$ to within a factor
$\gamma_p^2-\e$ would imply that $P=NP$. These results improve over
the earlier work~\cite{KNS10} which designed a polynomial time
algorithm for $M_p(A)$ whose approximation guarantee is
$(1+o(1))\gamma_p^2$ as $p\to \infty$, and which proved a
$\gamma_p^2-\e$ hardness results assuming the UGC rather than $P\neq
NP$.

The following Grothendieck-type inequality was proved in~\cite{NS09}
and independently in~\cite{GRSW10}. For every $n\times n$ matrix
$A=(a_{ij})$ and every $p\in [2,\infty)$ we have
\begin{equation}\label{eq:L_p gro}
\max_{\substack{x_1,\ldots,x_n\in \R^n\\\sum_{k=1}^n\|x_k\|_2^p\le
1}}\sum_{i=1}^n\sum_{j=1}^n a_{ij} \langle x_i,x_j\rangle\le
\gamma_p^2 \max_{\substack{t_1,\ldots,t_n\in
\R\\\sum_{k=1}^n|t_k|^p\le 1}}\sum_{i=1}^n\sum_{j=1}^n a_{ij} t_i
t_j.
\end{equation}
The constant $\gamma_p^2$ in~\eqref{eq:L_p gro} is sharp. The
validity of~\eqref{eq:L_p gro} implies that $M_p(A)$ can be computed
in polynomial time to within a factor $\gamma_p^2$. This follows
since the left hand side of~\eqref{eq:L_p gro} is the maximum of
$\sum_{i=1}^n\sum_{j=1}^n a_{ij} X_{ij}$, which is a linear
functional in the variables $(X_{ij})$, given the constraint that
$(X_{ij})$ is a symmetric positive semidefinite matrix and
$\sum_{i=1}^n X_{ii}^{p/2}\le 1$. The latter constraint is convex
since $p\ge 2$, and therefore this problem falls into the framework
of convex programming that was described in Section~\ref{sec:SDP}.
Thus the left hand side of~\eqref{eq:L_p gro} can be computed in
polynomial time with arbitrarily good precision.

Choosing the specific value $p=3$ in order to illustrate the current
satisfactory state of affairs concretely, the $NP$-hardness
threshold of computing $\max_{\sum_{i=1}^n |x_i|^3\le 1}
\sum_{i=1}^n\sum_{j=1}^n a_{ij}x_ix_j$ equals $2/\sqrt[3]{\pi}$.
Such a sharp $NP$-hardness result (with transcendental hardness
ratio) is quite remarkable, since it shows that the geometric
algorithm presented above probably yields the best possible
approximation guarantee even when one allows any polynomial time
algorithm whatsoever. Results of this type have been known to hold
under the UGC, but this $NP$-hardness result of~\cite{GRSW10} seems
to be the first time that such an algorithm for a simple to state
problem was shown to be optimal assuming $P\neq NP$.

When $p\in [1,2]$ one can easily show~\cite{NS09} that
\begin{equation}\label{eq:L_p gro p<2}
\max_{\substack{x_1,\ldots,x_n\in \R^n\\\sum_{k=1}^n\|x_k\|_2^p\le
1}}\sum_{i=1}^n\sum_{j=1}^n a_{ij} \langle x_i,x_j\rangle=
\max_{\substack{t_1,\ldots,t_n\in \R\\\sum_{k=1}^n|t_k|^p\le
1}}\sum_{i=1}^n\sum_{j=1}^n a_{ij} t_i t_j.
\end{equation}
While the identity~\eqref{eq:L_p gro p<2} seems to indicate the
problem of computing $M_p(A)$ in polynomial time might be easy for
$p\in (1,2)$, the above argument fails since the constraint
$\sum_{i=1}^n X_{ii}^{p/2}\le 1$ is no longer convex. This is
reflected by the fact that despite~\eqref{eq:L_p gro p<2} the
problem of $(1+\e)$-approximately  computing $M_1(A)$ is $W[1]$ complete~\cite{Alo06}. It
remains  open whether for $p\in (1,2)$ one can approximate  $M_p(A)$
in polynomial time up to a factor $O(1)$, and no hardness of
approximation result is known for this problem as well.

\begin{remark}\label{rem:PSD}
If $p\in [2,\infty]$ then for positive semidefinite matrices
$(a_{ij})$ the constant $\gamma_p^2$ in the right hand side
of~\eqref{eq:L_p gro} can be improved~\cite{NS09} to
$\gamma^{-2}_{p^*}$, where here and in what follows $p^*=p/(p-1)$.
For $p=\infty$ this estimate coincides with the classical
bound~\cite{Gro53,Rie74} that we have already encountered
in~\eqref{eq:PSD gro}, and it is sharp in the entire range $p\in
[2,\infty]$. Moreover, this bound shows that there exists a
polynomial time algorithm that takes as input a positive
semidefinite matrix $A$ and outputs a number that is guaranteed to
be within a factor $\gamma^{-2}_{p^*}$ of $M_p(A)$. Conversely, the
existence of a polynomial time algorithm for this problem whose
approximation guarantee is strictly smaller than $\gamma^{-2}_{p^*}$
would contradict the UGC~\cite{NS09}.
\end{remark}

\begin{remark}\label{rem:bilinear}
The bilinear variant of~\eqref{eq:L_p gro} is an immediate
consequence of the Grothendieck inequality~\eqref{eq:def gro}.
Specifically, assume that $p,q\in [1,\infty]$ and
$x_1,\ldots,x_m,y_1,\ldots,y_n\in \R^{m+n}$ satisfy
$\sum_{i=1}^m\|x_i\|_2^p\le 1$ and $\sum_{j=1}^n\|y_j\|_2^q\le 1$.
Write $\alpha_i=\|x_i\|_2$ and $\beta_j=\|y_j\|_2$. For an $m\times
n$ matrix $(a_{ij})$ the Grothendieck inequality provides
$\e_1,\ldots,\e_m,\delta_1,\ldots,\delta_n\in \{-1,1\}$ such that $
\sum_{i=1}^m\sum_{j=1}^n a_{ij} \langle x_i,y_j\rangle\le K_G
\sum_{i=1}^m\sum_{j=1}^n a_{ij}\alpha_i\beta_j\e_i\delta_j$. This
establishes the following inequality.
\begin{equation}\label{eq:pq-grothendieck}
\max_{\substack{\{x_i\}_{i=1}^m,\{y_j\}_{j=1}^n\subseteq
\R^{n+m}\\\sum_{i=1}^m \|x_i\|_2^p\le 1\\\sum_{j=1}^n\|y_j\|_2^q\le
1}}\sum_{i=1}^m\sum_{j=1}^n a_{ij}\langle x_i,y_j\rangle\le K_G\cdot
\max_{\substack{\{s_i\}_{i=1}^m,\{t_j\}_{j=1}^n\subseteq
\R\\\sum_{i=1}^m |s_i|^p\le 1\\\sum_{j=1}^n |t_j|^q\le
1}}\sum_{i=1}^m\sum_{j=1}^n a_{ij} s_it_j.
\end{equation}
Observe that the maximum on the right hand side
of~\eqref{eq:pq-grothendieck} is $\|A\|_{p\to q^*}$; the operator
norm of $A$ acting as a linear operator from $(\R^m,\|\cdot\|_p)$ to
$(\R^n,\|\cdot \|_{q^*})$. Moreover, if $p,q\ge 2$ then the left
hand side of~\eqref{eq:pq-grothendieck} can be computed in
polynomial time. Thus, for $p\ge 2\ge r\ge 1$, the generalized
Grothendieck inequality~\eqref{eq:pq-grothendieck} yields a
polynomial time algorithm that takes as input an $m\times n$ matrix
$A=(a_{ij})$ and outputs a number that is guaranteed to be within a
factor $K_G$ of $\|A\|_{p\to r}$. This algorithmic task has been
previously studied in~\cite{NWY00} (see
also~\cite[Sec.~4.3.2]{Nem07}), where for $p\ge 2\ge r\ge 1$ a
polynomial time algorithm was designed that approximates
$\|A\|_{p\to r}$ up to a factor $3\pi/\left(6\sqrt{3}-2\pi\right)\in
[2.293,2.294]$.  The above argument yields the approximation factor
$K_G<1.783$ as a formal consequence of the Grothendieck inequality.
The complexity of the problem of approximating $\|A\|_{p\to r}$ has
been studied in~\cite{BV11}, where it is shown that if either $p\ge
r>2$ or $2>p\ge r$ then it is $NP$-hard to approximate $\|A\|_{p\to
r}$ up to any constant factor, and unless $3$-colorability can be
solved in time $2^{(\log n)^{O(1)}}$, for any $\e\in (0,1)$ no
polynomial time algorithm can approximate $\|A\|_{p\to r}$ up to
$2^{(\log n)^{1-\e}}$.
\end{remark}

\begin{remark}\label{rem:convexity}
Let $K\subseteq \R^n$ be a compact and convex set which is invariant
under reflections with respect to the coordinate hyperplanes. Denote
by $C_K$ the smallest $C\in (0,\infty)$ such that for every $n\times
n$ matrix $(a_{ij})$ we have
\begin{equation}\label{eqK-grothendieck}
\max_{\substack{x_1,\ldots,x_n\in
\R^n\\(\|x_1\|_2,\ldots,\|x_n\|_2)\in K}}\sum_{i=1}^n\sum_{j=1}^n
a_{ij}\langle x_i,x_j\rangle\le C \max_{\substack{t_1,\ldots,t_n\in
\R\\(t_1,\ldots,t_n)\in K}}\sum_{i=1}^n\sum_{j=1}^n a_{ij} t_it_j.
\end{equation}
 Such generalized Grothendieck
inequalities are investigated in~\cite{NS09}, where bounds on $C_K$
are obtained under certain geometric assumptions on $K$. These
assumptions are easy to verify when $K=\{x\in \R^n:\ \|x\|_p\le
1\}$, yielding~\eqref{eq:L_p gro}. More subtle inequalities of this
type for other convex bodies $K$ are discussed in~\cite{NS09}, but
we will not describe them here. The natural bilinear version
of~\eqref{eqK-grothendieck} is: if $K\subseteq \R^m$ and $L\subseteq
\R^n$ are compact and convex sets that are invariant under
reflections with respect to the coordinate hyperplanes then let
$C_{K,L}$ denote the smallest constant $C\in (0,\infty)$ such that
for every $m\times n$ matrix $(a_{ij})$ we have
\begin{equation}\label{eq:KL-grothendieck}
\max_{\substack{\{x_i\}_{i=1}^m,\{y_j\}_{j=1}^n\subseteq
\R^{n+m}\\(\|x_1\|_2,\ldots,\|x_m\|_2)\in
K\\(\|y_1\|_2,\ldots,\|y_n\|_2)\in L}}\sum_{i=1}^m\sum_{j=1}^n
a_{ij}\langle x_i,y_j\rangle\le C
\max_{\substack{\{s_i\}_{i=1}^m,\{t_j\}_{j=1}^n\subseteq
\R\\(s_1,\ldots,s_m)\in K\\(t_1,\ldots,t_n)\in
L}}\sum_{i=1}^m\sum_{j=1}^n a_{ij} s_it_j.
\end{equation}
The argument in Remark~\ref{rem:bilinear} shows that $C_{K,L}\le
K_G$. Under certain geometric assumptions on $K,L$ this bound can be
improved~\cite{NS09}.
\end{remark}

\newcommand{\C}{\mathbb C}

\section{Higher rank Grothendieck inequalities}\label{sec:rank}

We have already seen several variants of the classical Grothendieck
inequality~\eqref{eq:def gro}, including the Grothendieck inequality
for graphs~\eqref{eq:def gro graph}, the variant of the positive
semidefinite Grothendieck inequality arising from the Kernel
Clustering problem~\eqref{eq:our gro with K intro}, and Grothendieck
inequalities for convex bodies other than the cube~\eqref{eq:L_p
gro}, \eqref{eq:pq-grothendieck}, \eqref{eqK-grothendieck},
\eqref{eq:KL-grothendieck}. The literature contains additional
 variants of the Grothendieck inequality, some of which
will be described in this section.

Let $G=(\{1,\ldots,n\},E)$ be a graph and fix $q,r\in \N$.
Following~\cite{BOV10}, let $K(q\to r,G)$ be the smallest constant
$K\in (0,\infty)$ such that for every $n\times n$ matrix
$A=(a_{ij})$ we have
\begin{equation}\label{eq:def gro rank}
\max_{x_1,\ldots,x_n\in S^{q-1}} \sum_{\substack{i,j\in
\{1,\ldots,n\}\\\{i,j\}\in E}}a_{ij}\langle x_i,x_j\rangle\le K
\max_{y_1,\ldots,y_n\in S^{r-1}} \sum_{\substack{i,j\in
\{1,\ldots,n\}\\\{i,j\}\in E}}a_{ij}\langle y_i,y_j\rangle .
\end{equation}
Set also $K(r,G)=\sup_{q\in \N} K(q\to r,G)$. We similarly define
$K^+(q\to r,G)$ to be the smallest constant $K\in (0,\infty)$
satisfying~\eqref{eq:def gro rank} for all positive semidefinite
matrices $A$, and correspondingly $K^+(r,G)=\sup_{q\in \N} K^+(q\to
r,G)$.

To link these definitions to what we have already seen in this
article, observe that $K_G$ is the supremum of $K(1,G)$ over all
finite bipartite graphs $G$, and due to the results described in
Section~\ref{sec:kernel} we have
\begin{equation}\label{eq:sup K^+}
\sup_{n\in \N}K^+\left(r,K_n^{\circlearrowleft} \right)=\sup_{n\in
\N}\sup_{x_1,\ldots,x_n\in S^{r-1}}\frac{1}{C\left(\langle
x_i,x_j\rangle)_{i,j=1}^n\right)},
\end{equation}
where $K_n^{\circlearrowleft}$ is the complete graph on $n$-vertices
with self loops. Recall that the definition of $C(B)$ for a positive
semidefinite matrix $B$ is given in the paragraph
following~\eqref{eq:our gro with K intro}.

The most important special case of~\eqref{eq:def gro rank} is when
$r=2$, since the  supremum of $K(2,G)$ over all finite bipartite
graphs $G$, denoted $K_G^\C$, is the complex Grothendieck constant,
a fundamental quantity whose value has been investigated
in~\cite{Gro53,LP68,Pisier78,Haa87,Kon90}. The best known bounds on
$K_G^\C$ are $1.338<K_G^\C< 1.4049$; see~\cite[Sec.~4]{Pis11} for
more information on this topic. We also refer to~\cite{Dav85,Ton86}
for information of the constants $K(2q\to 2,G)$ where $G$ is a
bipartite graph. The supremum of $K(q\to r,G)$ over all biparpite
graphs $G$ was investigated in~\cite{Kri79} for $r=1$ and
in~\cite{Kon90} for $r=2$; see also~\cite{Kon92} for a unified
treatment of these cases. The higher rank constants $K(q\to r,G)$
when $G$ is bipartite were introduced in~\cite{BBT09}.
Definition~\eqref{eq:def gro rank} in full generality is due
to~\cite{BOV10} where several estimates on $K(q\to r,G)$ are given.
One of the motivations of~\cite{BOV10} is the case $r=3$ (and $G$ a
subgraph of the grid $\Z^3$), based on the connection to the
polynomial time approximation of ground states of spin glasses as
described in Section~\ref{sec:spin}; the case $r=1$ was discussed in
Section~\ref{sec:spin} in connection with the Ising model, but the
case $r=3$ corresponds to the more physically realistic Heisenberg
model of vector-valued spins. The parameter $\sup_{n\in
\N}K^+\left(r,K_n^{\circlearrowleft} \right)$ (recall~\eqref{eq:sup
K^+}) was studied in~\cite{BBT09} in the context of quantum
information theory, and in~\cite{BOV10-1} it was shown that
\begin{equation}\label{eq:improve PSD}
K^+\left(1,K_n^{\circlearrowleft} \right)\le
\frac{\pi}{n}\left(\frac{\Gamma((n+1)/2)}{\Gamma(n/2)}\right)^2=\frac{\pi}{2}-\frac{\pi}{4n}+O\left(\frac{1}{n^2}\right),
\end{equation}
and
$$
\sup_{n\in \N}K^+\left(r,K_n^{\circlearrowleft}
\right)=\frac{r}{2}\left(\frac{\Gamma(r/2)}{\Gamma((r+1)/2)}\right)^2=1+\frac{1}{2r}+O\left(\frac{1}{r^2}\right).
$$
We refer to~\cite{BOV10-1} for a corresponding UGC hardness result. Note that~\eqref{eq:improve PSD} improves over~\eqref{eq:PSD gro}
for fixed $n\in \N$.
\section{Hardness of approximation}\label{sec:hardness}

We have seen examples of how Grothendieck-type inequalities yield
upper bounds on the best possible polynomial time approximation
ratio of certain optimization problems. From the algorithmic  and
computational complexity viewpoint it is interesting to prove
computational lower bounds as well, i.e., results that rule out the
existence of efficient algorithms achieving a certain approximation
guarantee. Such results are known as hardness or inapproximability
results, and as explained in Section~\ref{sec:complexity
assumptions}, at present  the state of the art allows one to prove
such results while relying on complexity theoretic assumptions such
as $P \not= NP$ or the Unique Games Conjecture. A nice feature of
the known hardness results for problems in which a Grothendieck-type
inequality has been applied is that often the hardness results
(lower bounds) exactly match the approximation ratios (upper
bounds). In this section we briefly review the known hardness
results for optimization problems associated with Grothendieck-type
inequalities.

Let $K_{n,n}$-${\sf QP}$ denote the optimization problem associated
with the classical Grothendieck inequality (the acronym ${\sf QP}$
stands for ``quadratic programming"). Thus, in the problem
$K_{n,n}$-${\sf QP}$ we are given an $n\times n$ real matrix
$(a_{ij})$ and the goal is to determine the quantity
$$
\max\left\{\sum_{i=1}^m\sum_{j=1}^n a_{ij} \e_i\delta_j:
\{\e_i\}_{i=1}^m,\{\delta_j\}_{j=1}^n\subseteq \{-1,1\}\right\}.
$$

As explained in~\cite{AN06}, the MAX DICUT problem can be framed as
a special case of the problem $K_{n,n}$-${\sf QP}$. Hence, as a
consequence of~\cite{Haa01}, we know that for every $\e\in (0,1)$,
assuming $P \not= NP$ there is no polynomial time algorithm that
approximates the $K_{n,n}$-${\sf QP}$ problem within ratio
$\frac{13}{12}-\e$. In~\cite{KO09} it is shown that the lower
bound~\eqref{eq:reeds} on the Grothendieck constant can be
translated into a hardness result, albeit relying on the Unique
Games Conjecture. Namely, letting $\eta_0$ be as
in~\eqref{eq:reeds}, for every $\e\in (0,1)$ assuming the UGC there
is no polynomial time algorithm that approximates the
$K_{n,n}$-${\sf QP}$ problem within a ratio
$\frac{\pi}{2}e^{\eta_0^2}-\e$.

We note that all the hardness results cited here rely on the
well-known paradigm of {\it dictatorship testing}. A lower bound on
the integrality gap of a semidefinite program, such as the estimate
$K_G\ge \frac{\pi}{2}e^{\eta_0^2}$, can be translated into a
probabilistic test to check whether a function $f: \{-1,1\}^n
\mapsto \{-1,1\}$ is a dictatorship, i.e., of the form $f(x)=x_i$
for some fixed $i\in \{1,\ldots,n\}$. If $f$ is indeed a
dictatorship, then the test passes with probability $c$ and if $f$
is ``far from a dictator" (in a formal sense that we do not describe
here), the test passes with probability at most $s$. The ratio $c/s$
 corresponds exactly to the UGC-based hardness lower bound. It is well-known how to
prove a UGC-based hardness result once we have the appropriate
dictatorship test; see the survey~\cite{Kho10}.

The above quoted result of~\cite{KO09} relied on explicitly knowing
the lower bound construction~\cite{Ree91} leading to the estimate
$K_G\ge \frac{\pi}{2}e^{\eta_0^2}$. On the other hand,
in~\cite{RS09}, building on the earlier work~\cite{Prasad}, it is
shown that {\it any} lower bound on the Grothedieck constant can be
translated into a UGC-based hardness result, even without explicitly
knowing the construction! Thus, modulo the UGC, the best polynomial
time algorithm to approximate the $K_{n,n}$-${\sf QP}$ problem is
via the Grothendieck inequality, even though we do not know the
precise value of $K_G$. Formally, for every $\e\in (0,1)$, assuming
the UGC there is no polynomial time algorithm that approximates the
$K_{n,n}$-${\sf QP}$ problem within a factor $K_G-\e$.

Let $K_{n,n}$-${\sf QP}_{\sf PSD}$ be the special case of the
$K_{n,n}$-${\sf QP}$ problem where the input matrix $(a_{ij})$ is
assumed to be positive semidefinite. By considering matrices that
are Laplacians of graphs one sees that the MAX CUT problem is a
special case of the problem $K_{n,n}$-${\sf QP}_{\sf PSD}$
(see~\cite{KN09}). Hence, due to~\cite{Haa01}, we know that for
every $\e\in (0,1)$, assuming $P \not= NP$ there is no polynomial
time algorithm that approximates the $K_{n,n}$-${\sf QP}_{\sf PSD}$
problem within ratio $\frac{17}{16}-\e$. Moreover, it is proved
in~\cite{KN09} that for every $\e\in (0,1)$, assuming the UGC there
is no polynomial time algorithm that approximates the
$K_{n,n}$-${\sf QP}_{\sf PSD}$ problem within ratio
$\frac{\pi}{2}-\e$, an optimal hardness result due to the positive
semidefinite Grothendieck inequality~\eqref{eq:PSD gro}. This
follows from the more general results for the Kernel Clustering
problem described later.

Let $(a_{ij})$ be an $n\times n$ real matrix with zeroes on the
diagonal. The $K_n$-${\sf QP}$
 problem seeks to determine the quantity
$$
\max\left\{\sum_{i=1}^m\sum_{j=1}^n a_{ij} \e_i\e_j:
\{\e_i\}_{i=1}^m\subseteq \{-1,1\}\right\}.
$$
In~\cite{KS11} it is proved that for every $\gamma\in (0,1/6)$,
assuming that $NP$ does not have a $2^{(\log n)^{O(1)}}$ time
deterministic algorithm, there is no polynomial time algorithm that
approximates the $K_n$-${\sf QP}$ problem within ratio $(\log
n)^\gamma$. This improves over~\cite{ABKHS05} where a hardness
factor of $(\log n)^c$ was proved, under the same complexity
assumption, for an unspecified universal constant $c>0$. Recall
that, as explained in Section~\ref{sec:graphs}, there is an
algorithm for $K_n$-${\sf QP}$ that achieves a ratio of $O(\log n)$,
so there remains an asymptotic gap in our understanding of the
complexity of the $K_n$-${\sf QP}$ problem. For the maximum acyclic
subgraph problem, as discussed in Section~\ref{sec:acyclic}, the gap
between the upper and lower bounds is even larger. We have already
seen that an approximation factor of $O(\log n)$ is achievable, but
from the hardness perspective we know due to~\cite{PY91} that there
exists $\e_0>0$ such that assuming $P\neq NP$ there is no polynomial
time algorithm for the maximum acyclic subgraph problem that
achieves an approximation ratio less than $1+\e_0$. In~\cite{GMR08}
it was shown that assuming the UGC there is no polynomial time
algorithm for the maximum acyclic subgraph problem that achieves any
constant approximation ratio.

Fix $p\in (0,\infty)$. As discussed in Section~\ref{sec:L_p}, the
$L_p$ Grothendieck problem is as follows. Given an $n\times n$ real
matrix $A=(a_{ij})$ with zeros on the diagonal, the goal is to
determine the quantity $M_p(A)$ defined in~\eqref{eq:def: M_p}. For
$p\in (2,\infty)$ it was shown in~\cite{GRSW10} that for every
$\e\in (0,1)$, assuming $P \not= NP$ there is no polynomial time
algorithm that approximates the $L_p$ Grothendieck problem  within a
ratio $\gamma_p^2 -\e$. Here $\gamma_p$ is defined as
in~\eqref{eq:def gamma_p}. This result (nontrivially) builds on  the
previous result of~\cite{KNS10} that obtained the same conclusion
while assuming the UGC rather than $P\neq NP$.

For the Kernel Clustering problem with a $k \times k$ hypothesis
matrix $B$, an optimal hardness result is obtained in~\cite{KN10} in
terms of the parameters $R(B)$ and $C(B)$ described in
Section~\ref{sec:kernel}. Specifically for a fixed $k\times k$
symmetric positive semidefinite matrix $B$ and for every $\e\in
(0,1)$, assuming the UGC there is no polynomial time algorithm that,
given an $n\times n$ matrix $A$ approximates the quantity
$\Clust(A|B)$  within ratio $\frac{R(B)^2}{C(B)}-\e$. When $B = I_k$
is the $k \times k$ identity matrix, the following hardness result
is obtained in~\cite{KN09}. Let $\e > 0$ be an arbitrarily small
constant. Assuming the UGC, there is no polynomial time algorithm
that approximates $\Clust(A|I_2)$ within ratio $\frac{\pi}{2} - \e$.
Similarly, assuming the UGC there is no polynomial time algorithm
that approximates $\Clust(A|I_3)$ within ratio $\frac{16
\pi}{27}-\e$, and, using also the solution of the propeller
conjecture in $\R^3$ given in~\cite{HJN11}, there is no polynomial
time algorithm that approximates $\Clust(A|I_4)$ within ratio
$\frac{2 \pi}{3}-\e$. Furthermore, for $k\ge 5$, assuming the
propeller conjecture and the UGC, there is no polynomial time
algorithm that approximates $\Clust(A|I_k)$ within ratio $\frac{8
\pi}{9} \left( 1- \frac{1}{k} \right)-\e$.

\bigskip

\noindent{\bf Acknowledgements.} We are grateful to Oded Regev for
many helpful suggestions.

\bibliographystyle{abbrv}
\bibliography{Grothendieck-CPAM}

\def\polhk#1{\setbox0=\hbox{#1}{\ooalign{\hidewidth
  \lower1.5ex\hbox{`}\hidewidth\crcr\unhbox0}}} \def\cprime{$'$}
  \def\cprime{$'$}
\begin{thebibliography}{100}

\bibitem{AGT06}
A.~Ac{\'{\i}}n, N.~Gisin, and B.~Toner.
\newblock Grothendieck's constant and local models for noisy entangled quantum
  states.
\newblock {\em Phys. Rev. A (3)}, 73(6, part A):062105, 5, 2006.

\bibitem{Alo06}
N.~Alon.
\newblock Maximizing a quadratic form on the $\ell_1^n$ unit ball.
\newblock Unpublished manuscript, 2006.

\bibitem{AB08}
N.~Alon and E.~Berger.
\newblock The {G}rothendieck constant of random and pseudo-random graphs.
\newblock {\em Discrete Optim.}, 5(2):323--327, 2008.

\bibitem{ACHKRS10}
N.~Alon, A.~Coja-Oghlan, H.~H{\`a}n, M.~Kang, V.~R{\"o}dl, and M.~Schacht.
\newblock Quasi-randomness and algorithmic regularity for graphs with general
  degree distributions.
\newblock {\em SIAM J. Comput.}, 39(6):2336--2362, 2010.

\bibitem{ADLRY94}
N.~Alon, R.~A. Duke, H.~Lefmann, V.~R{\"o}dl, and R.~Yuster.
\newblock The algorithmic aspects of the regularity lemma.
\newblock {\em J. Algorithms}, 16(1):80--109, 1994.

\bibitem{AFKK03}
N.~Alon, W.~Fernandez de~la Vega, R.~Kannan, and M.~Karpinski.
\newblock Random sampling and approximation of {MAX}-{CSP}s.
\newblock {\em J. Comput. System Sci.}, 67(2):212--243, 2003.
\newblock Special issue on STOC2002 (Montreal, QC).

\bibitem{AMMN06}
N.~Alon, K.~Makarychev, Y.~Makarychev, and A.~Naor.
\newblock Quadratic forms on graphs.
\newblock {\em Invent. Math.}, 163(3):499--522, 2006.

\bibitem{AN06}
N.~Alon and A.~Naor.
\newblock Approximating the cut-norm via {G}rothendieck's inequality.
\newblock {\em SIAM J. Comput.}, 35(4):787--803 (electronic), 2006.

\bibitem{AO95}
N.~Alon and A.~Orlitsky.
\newblock Repeated communication and {R}amsey graphs.
\newblock {\em IEEE Trans. Inform. Theory}, 41(5):1276--1289, 1995.

\bibitem{AS00}
N.~Alon and J.~H. Spencer.
\newblock {\em The probabilistic method}.
\newblock Wiley-Interscience Series in Discrete Mathematics and Optimization.
  Wiley-Interscience [John Wiley \& Sons], New York, second edition, 2000.
\newblock With an appendix on the life and work of Paul Erd\H os.

\bibitem{AB09}
S.~Arora and B.~Barak.
\newblock {\em Computational complexity}.
\newblock Cambridge University Press, Cambridge, 2009.
\newblock A modern approach.

\bibitem{ABKHS05}
S.~Arora, E.~Berger, G.~Kindler, E.~Hazan, and S.~Safra.
\newblock On non-approximability for quadratic programs.
\newblock In {\em 46th Annual Symposium on Foundations of Computer Science},
  pages 206--215. IEEE Computer Society, 2005.

\bibitem{Bac84}
C.~P. Bachas.
\newblock Computer-intractability of the frustration model of a spin glass.
\newblock {\em J. Phys. A}, 17(13):L709--L712, 1984.

\bibitem{BBC02}
N.~Bansal, A.~Blum, and S.~Chawla.
\newblock Correlation clustering.
\newblock In {\em 43rd Annual IEEE Symposium on Foundations of Computer
  Science}, pages 238--247, 2002.

\bibitem{BW09}
N.~Bansal and R.~Williams.
\newblock Regularity lemmas and combinatorial algorithms.
\newblock In {\em 2009 50th {A}nnual {IEEE} {S}ymposium on {F}oundations of
  {C}omputer {S}cience ({FOCS} 2009)}, pages 745--754. IEEE Computer Soc., Los
  Alamitos, CA, 2009.

\bibitem{Bar82}
F.~Barahona.
\newblock On the computational complexity of {I}sing spin glass models.
\newblock {\em J. Phys. A}, 15(10):3241--3253, 1982.

\bibitem{BV11}
A.~Bhaskara and A.~Vijayaraghavan.
\newblock Approximating matrix $p$-norms.
\newblock In {\em Proceedings of the Twenty-Second Annual ACM-SIAM Symposium on
  Discrete Algorithms}, pages 497--511, 2011.

\bibitem{BMRU80}
I.~Bieche, R.~Maynard, R.~Rammal, and J.-P. Uhry.
\newblock On the ground states of the frustration model of a spin glass by a
  matching method of graph theory.
\newblock {\em J. Phys. A}, 13(8):2553--2576, 1980.

\bibitem{Ble01}
R.~Blei.
\newblock {\em Analysis in integer and fractional dimensions}, volume~71 of
  {\em Cambridge Studies in Advanced Mathematics}.
\newblock Cambridge University Press, Cambridge, 2001.

\bibitem{Ble79}
R.~C. Blei.
\newblock Multidimensional extensions of the {G}rothendieck inequality and
  applications.
\newblock {\em Ark. Mat.}, 17(1):51--68, 1979.

\bibitem{BMMN11}
M.~Braverman, K.~Makarychev, Y.~Makarychev, and A.~Naor.
\newblock The {G}rothendieck constant is strictly smaller than {K}rivine's
  bound.
\newblock An extended abstract will appear in 52nd Annual IEEE Symposium on
  Foundations of Computer Science. Preprint available at
  \url{http://arxiv.org/abs/1103.6161}, 2011.

\bibitem{BBT09}
J.~Briet, H.~Buhrman, and B.~Toner.
\newblock A generalized {G}rothendieck inequality and entanglement in {X}{O}{R}
  games.
\newblock Preprint available at \url{http://arxiv.org/abs/0901.2009}, 2009.

\bibitem{BOV10}
J.~Briet, F.~M. de~Oliveira~Filho, and V.~F.
\newblock Grothendieck inequalities for semidefinite programs with rank
  constraint.
\newblock Preprint available at \url{http://arxiv.org/abs/1011.1754}, 2010.

\bibitem{BOV10-1}
J.~Bri{\"e}t, F.~M. de~Oliveira~Filho, and F.~Vallentin.
\newblock The positive semidefinite {G}rothendieck problem with rank
  constraint.
\newblock In {\em Automata, Languages and Programming, 37th International
  Colloquium, Part I}, pages 31--42, 2010.

\bibitem{CGV05}
M.~Charikar, V.~Guruswami, and A.~Wirth.
\newblock Clustering with qualitative information.
\newblock {\em J. Comput. System Sci.}, 71(3):360--383, 2005.

\bibitem{CMM07}
M.~Charikar, K.~Makarychev, and Y.~Makarychev.
\newblock On the advantage over random for maximum acyclic subgraph.
\newblock In {\em 48th Annual IEEE Symposium on Foundations of Computer
  Science}, pages 625--633, 2007.

\bibitem{CW04}
M.~Charikar and A.~Wirth.
\newblock Maximizing quadratic programs: extending {G}rothendieck's inequality.
\newblock In {\em 45th Annual Symposium on Foundations of Computer Science},
  pages 54--60. IEEE Computer Society, 2004.

\bibitem{CHTW04}
R.~Cleve, P.~H{\o}yer, B.~Toner, and J.~Watrous.
\newblock Consequences and limits of nonlocal strategies.
\newblock In {\em 19th Annual IEEE Conference on Computational Complexity},
  pages 236--249, 2004.

\bibitem{CCF10}
A.~Coja-Oghlan, C.~Cooper, and A.~Frieze.
\newblock An efficient sparse regularity concept.
\newblock {\em SIAM J. Discrete Math.}, 23(4):2000--2034, 2009/10.

\bibitem{CF11}
D.~Conlon and J.~Fox.
\newblock Bounds for graph regularity and removal lemmas.
\newblock Preprint available at \url{http://arxiv.org/abs/1107.4829}, 2011.

\bibitem{Coo06}
S.~Cook.
\newblock The {P} versus {NP} problem.
\newblock In {\em The millennium prize problems}, pages 87--104. Clay Math.
  Inst., Cambridge, MA, 2006.

\bibitem{Dav85}
A.~M. Davie.
\newblock Matrix norms related to {G}rothendieck's inequality.
\newblock In {\em Banach spaces ({C}olumbia, {M}o., 1984)}, volume 1166 of {\em
  Lecture Notes in Math.}, pages 22--26. Springer, Berlin, 1985.

\bibitem{DFS08}
J.~Diestel, J.~H. Fourie, and J.~Swart.
\newblock {\em The metric theory of tensor products}.
\newblock American Mathematical Society, Providence, RI, 2008.
\newblock Grothendieck's r{\'e}sum{\'e} revisited.

\bibitem{DJT95}
J.~Diestel, H.~Jarchow, and A.~Tonge.
\newblock {\em Absolutely summing operators}, volume~43 of {\em Cambridge
  Studies in Advanced Mathematics}.
\newblock Cambridge University Press, Cambridge, 1995.

\bibitem{DF99}
R.~G. Downey and M.~R. Fellows.
\newblock {\em Parameterized complexity}.
\newblock Monographs in Computer Science. Springer-Verlag, New York, 1999.

\bibitem{ER60}
P.~Erd{\H{o}}s and A.~R{\'e}nyi.
\newblock On the evolution of random graphs.
\newblock {\em Magyar Tud. Akad. Mat. Kutat\'o Int. K\"ozl.}, 5:17--61, 1960.

\bibitem{FR94}
P.~C. Fishburn and J.~A. Reeds.
\newblock Bell inequalities, {G}rothendieck's constant, and root two.
\newblock {\em SIAM J. Discrete Math.}, 7(1):48--56, 1994.

\bibitem{FK99}
A.~Frieze and R.~Kannan.
\newblock Quick approximation to matrices and applications.
\newblock {\em Combinatorica}, 19(2):175--220, 1999.

\bibitem{GJ79}
M.~R. Garey and D.~S. Johnson.
\newblock {\em Computers and intractability}.
\newblock W. H. Freeman and Co., San Francisco, Calif., 1979.
\newblock A guide to the theory of NP-completeness, A Series of Books in the
  Mathematical Sciences.

\bibitem{Gar07}
D.~J.~H. Garling.
\newblock {\em Inequalities: a journey into linear analysis}.
\newblock Cambridge University Press, Cambridge, 2007.

\bibitem{GS07}
S.~Gerke and A.~Steger.
\newblock A characterization for sparse {$\epsilon$}-regular pairs.
\newblock {\em Electron. J. Combin.}, 14(1):Research Paper 4, 12 pp.
  (electronic), 2007.

\bibitem{GW95}
M.~X. Goemans and D.~P. Williamson.
\newblock Improved approximation algorithms for maximum cut and satisfiability
  problems using semidefinite programming.
\newblock {\em J. Assoc. Comput. Mach.}, 42(6):1115--1145, 1995.

\bibitem{GV96}
G.~H. Golub and C.~F. Van~Loan.
\newblock {\em Matrix computations}.
\newblock Johns Hopkins Studies in the Mathematical Sciences. Johns Hopkins
  University Press, Baltimore, MD, third edition, 1996.

\bibitem{Gow97}
W.~T. Gowers.
\newblock Lower bounds of tower type for {S}zemer\'edi's uniformity lemma.
\newblock {\em Geom. Funct. Anal.}, 7(2):322--337, 1997.

\bibitem{Gro53}
A.~Grothendieck.
\newblock R\'esum\'e de la th\'eorie m\'etrique des produits tensoriels
  topologiques.
\newblock {\em Bol. Soc. Mat. S\~ao Paulo}, 8:1--79, 1953.

\bibitem{GLS93}
M.~Gr{\"o}tschel, L.~Lov{\'a}sz, and A.~Schrijver.
\newblock {\em Geometric algorithms and combinatorial optimization}, volume~2
  of {\em Algorithms and Combinatorics}.
\newblock Springer-Verlag, Berlin, second edition, 1993.

\bibitem{GMR08}
V.~Guruswami, R.~Manokaran, and P.~Raghavendra.
\newblock Beating the random ordering is hard: Inapproximability of maximum
  acyclic subgraph.
\newblock In {\em 49th Annual IEEE Symposium on Foundations of Computer
  Science}, pages 573--582, 2008.

\bibitem{GRSW10}
V.~Guruswami, P.~Raghavendra, R.~Saket, and Y.~Wu.
\newblock Bypassing {U}{G}{C} from some optimal geometric inapproximability
  results.
\newblock {\em Electronic Colloquium on Computational Complexity (ECCC)},
  17:177, 2010.

\bibitem{Haa85}
U.~Haagerup.
\newblock The {G}rothendieck inequality for bilinear forms on
  {$C^\ast$}-algebras.
\newblock {\em Adv. in Math.}, 56(2):93--116, 1985.

\bibitem{Haa87}
U.~Haagerup.
\newblock A new upper bound for the complex {G}rothendieck constant.
\newblock {\em Israel J. Math.}, 60(2):199--224, 1987.

\bibitem{Haa01}
J.~H{\aa}stad.
\newblock Some optimal inapproximability results.
\newblock {\em J. ACM}, 48(4):798--859 (electronic), 2001.

\bibitem{HV04}
J.~H{\aa}stad and S.~Venkatesh.
\newblock On the advantage over a random assignment.
\newblock {\em Random Structures Algorithms}, 25(2):117--149, 2004.

\bibitem{HJN11}
S.~Heilman, A.~Jagannath, and A.~Naor.
\newblock Solution of the propeller conjecture in $\mathbb{R}^3$.
\newblock Manuscript, 2011.

\bibitem{Hey06}
H.~Heydari.
\newblock Quantum correlation and {G}rothendieck's constant.
\newblock {\em J. Phys. A}, 39(38):11869--11875, 2006.

\bibitem{Jam87}
G.~J.~O. Jameson.
\newblock {\em Summing and nuclear norms in {B}anach space theory}, volume~8 of
  {\em London Mathematical Society Student Texts}.
\newblock Cambridge University Press, Cambridge, 1987.

\bibitem{JL01}
W.~B. Johnson and J.~Lindenstrauss.
\newblock Basic concepts in the geometry of {B}anach spaces.
\newblock In {\em Handbook of the geometry of {B}anach spaces, {V}ol. {I}},
  pages 1--84. North-Holland, Amsterdam, 2001.

\bibitem{JN76}
D.~B. Judin and A.~S. Nemirovski{\u\i}.
\newblock Informational complexity and effective methods for the solution of
  convex extremal problems.
\newblock {\em \`Ekonom. i Mat. Metody}, 12(2):357--369, 1976.

\bibitem{Juh82}
F.~Juh{\'a}sz.
\newblock The asymptotic behaviour of {L}ov\'asz' {$\theta $} function for
  random graphs.
\newblock {\em Combinatorica}, 2(2):153--155, 1982.

\bibitem{KMS98}
D.~Karger, R.~Motwani, and M.~Sudan.
\newblock Approximate graph coloring by semidefinite programming.
\newblock {\em J. ACM}, 45(2):246--265, 1998.

\bibitem{KS03}
B.~S. Kashin and S.~{\u{I}}. Sharek.
\newblock On the {G}ram matrices of systems of uniformly bounded functions.
\newblock {\em Tr. Mat. Inst. Steklova}, 243(Funkts. Prostran., Priblizh.,
  Differ. Uravn.):237--243, 2003.

\bibitem{KKMTV08}
J.~Kempe, H.~Kobayashi, K.~Matsumoto, B.~Toner, and T.~Vidick.
\newblock Entangled games are hard to approximate.
\newblock In {\em 49th Annual IEEE Symposium on Foundations of Computer
  Science}, pages 447--456, 2008.

\bibitem{Khot02}
S.~Khot.
\newblock On the power of unique 2-prover 1-round games.
\newblock In {\em Proceedings of the Thirty-Fourth Annual ACM Symposium on
  Theory of Computing}, pages 767--775 (electronic), New York, 2002. ACM.

\bibitem{Kho10}
S.~Khot.
\newblock On the unique games conjecture (invited survey).
\newblock In {\em Proceedings of the 25th Annual IEEE Conference on
  Computational Complexity}, pages 99--121, 2010.

\bibitem{KKMO07}
S.~Khot, G.~Kindler, E.~Mossel, and R.~O'Donnell.
\newblock Optimal inapproximability results for {MAX}-{CUT} and other
  2-variable {CSP}s?
\newblock {\em SIAM J. Comput.}, 37(1):319--357 (electronic), 2007.

\bibitem{KN08}
S.~Khot and A.~Naor.
\newblock Linear equations modulo 2 and the {$L_1$} diameter of convex bodies.
\newblock {\em SIAM J. Comput.}, 38(4):1448--1463, 2008.

\bibitem{KN09}
S.~Khot and A.~Naor.
\newblock Approximate kernel clustering.
\newblock {\em Mathematika}, 55(1-2):129--165, 2009.

\bibitem{KN10}
S.~Khot and A.~Naor.
\newblock Sharp kernel clustering algorithms and their associated
  {G}rothendieck inequalities.
\newblock In {\em Proceedings of the Twenty-First Annual ACM-SIAM Symposium on
  Discrete Algorithms}, pages 664--683, 2010.
\newblock Full version to appear in {R}andom {S}tructures and {A}lgorithms.

\bibitem{KO09}
S.~Khot and R.~O'Donnell.
\newblock S{DP} gaps and {UGC}-hardness for max-cut-gain.
\newblock {\em Theory Comput.}, 5:83--117, 2009.

\bibitem{KS11}
S.~Khot and S.~Safra.
\newblock A two prover one round game with strong soundness.
\newblock To appear in 52nd Annual IEEE Symposium on Foundations of Computer
  Science, 2011.

\bibitem{KNS10}
G.~Kindler, A.~Naor, and G.~Schechtman.
\newblock The {UGC} hardness threshold of the {$L_p$} {G}rothendieck problem.
\newblock {\em Math. Oper. Res.}, 35(2):267--283, 2010.

\bibitem{Koh97}
Y.~Kohayakawa.
\newblock Szemer\'edi's regularity lemma for sparse graphs.
\newblock In {\em Foundations of computational mathematics ({R}io de {J}aneiro,
  1997)}, pages 216--230. Springer, Berlin, 1997.

\bibitem{KR03}
Y.~Kohayakawa and V.~R{\"o}dl.
\newblock Szemer\'edi's regularity lemma and quasi-randomness.
\newblock In {\em Recent advances in algorithms and combinatorics}, volume~11
  of {\em CMS Books Math./Ouvrages Math. SMC}, pages 289--351. Springer, New
  York, 2003.

\bibitem{KRT03}
Y.~Kohayakawa, V.~R{\"o}dl, and L.~Thoma.
\newblock An optimal algorithm for checking regularity.
\newblock {\em SIAM J. Comput.}, 32(5):1210--1235 (electronic), 2003.

\bibitem{Kon90}
H.~K{\"o}nig.
\newblock On the complex {G}rothendieck constant in the {$n$}-dimensional case.
\newblock In {\em Geometry of {B}anach spaces ({S}trobl, 1989)}, volume 158 of
  {\em London Math. Soc. Lecture Note Ser.}, pages 181--198. Cambridge Univ.
  Press, Cambridge, 1990.

\bibitem{Kon92}
H.~K{\"o}nig.
\newblock Some remarks on the {G}rothendieck inequality.
\newblock In {\em General inequalities, 6 ({O}berwolfach, 1990)}, volume 103 of
  {\em Internat. Ser. Numer. Math.}, pages 201--206. Birkh\"auser, Basel, 1992.

\bibitem{Kon00}
H.~K{\"o}nig.
\newblock On an extremal problem originating in questions of unconditional
  convergence.
\newblock In {\em Recent progress in multivariate approximation
  ({W}itten-{B}ommerholz, 2000)}, volume 137 of {\em Internat. Ser. Numer.
  Math.}, pages 185--192. Birkh\"auser, Basel, 2001.

\bibitem{Krivine77}
J.-L. Krivine.
\newblock Sur la constante de {G}rothendieck.
\newblock {\em C. R. Acad. Sci. Paris S\'er. A-B}, 284(8):A445--A446, 1977.

\bibitem{Kri79}
J.-L. Krivine.
\newblock Constantes de {G}rothendieck et fonctions de type positif sur les
  sph\`eres.
\newblock {\em Adv. in Math.}, 31(1):16--30, 1979.

\bibitem{LV11}
M.~Laurent and A.~Varvitsiotis.
\newblock Computing the {G}rothendieck constant of some graph classes.
\newblock Preprint available at \url{http://arxiv.org/abs/1106.2735}, 2011.

\bibitem{LSS09}
T.~Lee, G.~Schechtman, and A.~Shraibman.
\newblock Lower bounds on quantum multiparty communication complexity.
\newblock In {\em Proceedings of the 24th Annual IEEE Conference on
  Computational Complexity}, pages 254--262, 2009.

\bibitem{LeS07}
T.~Lee and A.~Shraibman.
\newblock Lower bounds in communication complexity.
\newblock {\em Found. Trends Theor. Comput. Sci.}, 3(4):front matter, 263--399
  (2009), 2007.

\bibitem{LO96}
A.~S. Lewis and M.~L. Overton.
\newblock Eigenvalue optimization.
\newblock In {\em Acta numerica, 1996}, volume~5 of {\em Acta Numer.}, pages
  149--190. Cambridge Univ. Press, Cambridge, 1996.

\bibitem{LP68}
J.~Lindenstrauss and A.~Pe{\l}czy{\'n}ski.
\newblock Absolutely summing operators in {$L\sb{p}$}-spaces and their
  applications.
\newblock {\em Studia Math.}, 29:275--326, 1968.

\bibitem{LMSS07}
N.~Linial, S.~Mendelson, G.~Schechtman, and A.~Shraibman.
\newblock Complexity measures of sign matrices.
\newblock {\em Combinatorica}, 27(4):439--463, 2007.

\bibitem{LS09-1}
N.~Linial and A.~Shraibman.
\newblock Learning complexity vs. communication complexity.
\newblock {\em Combin. Probab. Comput.}, 18(1-2):227--245, 2009.

\bibitem{LS09}
N.~Linial and A.~Shraibman.
\newblock Lower bounds in communication complexity based on factorization
  norms.
\newblock {\em Random Structures Algorithms}, 34(3):368--394, 2009.

\bibitem{Lov79}
L.~Lov{\'a}sz.
\newblock On the {S}hannon capacity of a graph.
\newblock {\em IEEE Trans. Inform. Theory}, 25(1):1--7, 1979.

\bibitem{LP86}
L.~Lov{\'a}sz and M.~D. Plummer.
\newblock {\em Matching theory}, volume 121 of {\em North-Holland Mathematics
  Studies}.
\newblock North-Holland Publishing Co., Amsterdam, 1986.
\newblock Annals of Discrete Mathematics, 29.

\bibitem{LS07}
L.~Lov{\'a}sz and B.~Szegedy.
\newblock Szemer\'edi's lemma for the analyst.
\newblock {\em Geom. Funct. Anal.}, 17(1):252--270, 2007.

\bibitem{Mat70}
D.~W. Matula.
\newblock On the complete subgraphs of a random graph.
\newblock In {\em Proc. {S}econd {C}hapel {H}ill {C}onf. on {C}ombinatorial
  {M}athematics and its {A}pplications ({U}niv. {N}orth {C}arolina, {C}hapel
  {H}ill, {N}.{C}., 1970)}, pages 356--369. Univ. North Carolina, Chapel Hill,
  N.C., 1970.

\bibitem{Meg01}
A.~Megretski.
\newblock Relaxations of quadratic programs in operator theory and system
  analysis.
\newblock In {\em Systems, approximation, singular integral operators, and
  related topics (Bordeaux, 2000)}, volume 129 of {\em Oper. Theory Adv.
  Appl.}, pages 365--392. Birkh\"auser, Basel, 2001.

\bibitem{NS09}
A.~Naor and G.~Schechtman.
\newblock An approximation scheme for quadratic form maximization on convex
  bodies.
\newblock Manuscript, 2009.

\bibitem{Nem07}
A.~Nemirovski.
\newblock Advances in convex optimization: conic programming.
\newblock In {\em International {C}ongress of {M}athematicians. {V}ol. {I}},
  pages 413--444. Eur. Math. Soc., Z\"urich, 2007.

\bibitem{NRT99}
A.~Nemirovski, C.~Roos, and T.~Terlaky.
\newblock On maximization of quadratic form over intersection of ellipsoids
  with common center.
\newblock {\em Math. Program.}, 86(3, Ser. A):463--473, 1999.

\bibitem{Nes98}
Y.~Nesterov.
\newblock Semidefinite relaxation and nonconvex quadratic optimization.
\newblock {\em Optim. Methods Softw.}, 9(1-3):141--160, 1998.

\bibitem{NWY00}
Y.~Nesterov, H.~Wolkowicz, and Y.~Ye.
\newblock Semidefinite programming relaxations of nonconvex quadratic
  optimization.
\newblock In {\em Handbook of semidefinite programming}, volume~27 of {\em
  Internat. Ser. Oper. Res. Management Sci.}, pages 361--419. Kluwer Acad.
  Publ., Boston, MA, 2000.

\bibitem{PY91}
C.~H. Papadimitriou and M.~Yannakakis.
\newblock Optimization, approximation, and complexity classes.
\newblock {\em J. Comput. System Sci.}, 43(3):425--440, 1991.

\bibitem{PWPVJ08}
D.~P{\'e}rez-Garc{\'{\i}}a, M.~M. Wolf, C.~Palazuelos, I.~Villanueva, and
  M.~Junge.
\newblock Unbounded violation of tripartite {B}ell inequalities.
\newblock {\em Comm. Math. Phys.}, 279(2):455--486, 2008.

\bibitem{Pisier78}
G.~Pisier.
\newblock Grothendieck's theorem for noncommutative {$C\sp{\ast} $}-algebras,
  with an appendix on {G}rothendieck's constants.
\newblock {\em J. Funct. Anal.}, 29(3):397--415, 1978.

\bibitem{Pis86}
G.~Pisier.
\newblock {\em Factorization of linear operators and geometry of {B}anach
  spaces}, volume~60 of {\em CBMS Regional Conference Series in Mathematics}.
\newblock Published for the Conference Board of the Mathematical Sciences,
  Washington, DC, 1986.

\bibitem{Pis11}
G.~Pisier.
\newblock Grothendieck's theorem, past and present.
\newblock Preprint available at \url{http://arxiv.org/abs/1101.4195}, 2011.

\bibitem{Pit08}
I.~Pitowsky.
\newblock New {B}ell inequalities for the singlet state: going beyond the
  {G}rothendieck bound.
\newblock {\em J. Math. Phys.}, 49(1):012101, 11, 2008.

\bibitem{Prasad}
P.~Raghavendra.
\newblock Optimal algorithms and inapproximability results for every {C}{S}{P}?
\newblock In {\em Proceedings of the 40th Annual ACM Symposium on Theory of
  Computing}, pages 245--254, 2008.

\bibitem{RS09}
P.~Raghavendra and D.~Steurer.
\newblock Towards computing the {G}rothendieck constant.
\newblock In {\em Proceedings of the Twentieth Annual ACM-SIAM Symposium on
  Discrete Algorithms}, pages 525--534, 2009.

\bibitem{Ree91}
J.~A. Reeds.
\newblock A new lower bound on the real {G}rothendieck constant.
\newblock Unpublished manuscript, available at
  \url{http://www.dtc.umn.edu/reedsj/bound2.dvi}, 1991.

\bibitem{RT09}
O.~Regev and B.~Toner.
\newblock Simulating quantum correlations with finite communication.
\newblock {\em SIAM J. Comput.}, 39(4):1562--1580, 2009/10.

\bibitem{Rie74}
R.~E. Rietz.
\newblock A proof of the {G}rothendieck inequality.
\newblock {\em Israel J. Math.}, 19:271--276, 1974.

\bibitem{Sip97}
M.~Sipser.
\newblock {\em Introduction to the theory of computation}.
\newblock PWS Publishing Company, 1997.

\bibitem{Smi88}
R.~R. Smith.
\newblock Completely bounded multilinear maps and {G}rothendieck's inequality.
\newblock {\em Bull. London Math. Soc.}, 20(6):606--612, 1988.

\bibitem{SSGB07}
L.~Song, A.~Smola, A.~Gretton, and K.~A. Borgwardt.
\newblock A dependence maximization view of clustering.
\newblock In {\em Proceedings of the 24th international conference on Machine
  learning}, pages 815 -- 822, 2007.

\bibitem{Sze78}
E.~Szemer{\'e}di.
\newblock Regular partitions of graphs.
\newblock In {\em Probl\`emes combinatoires et th\'eorie des graphes ({C}olloq.
  {I}nternat. {CNRS}, {U}niv. {O}rsay, {O}rsay, 1976)}, volume 260 of {\em
  Colloq. Internat. CNRS}, pages 399--401. CNRS, Paris, 1978.

\bibitem{Ton78}
A.~Tonge.
\newblock The von {N}eumann inequality for polynomials in several
  {H}ilbert-{S}chmidt operators.
\newblock {\em J. London Math. Soc. (2)}, 18(3):519--526, 1978.

\bibitem{Ton86}
A.~Tonge.
\newblock The complex {G}rothendieck inequality for {$2\times 2$} matrices.
\newblock {\em Bull. Soc. Math. Gr\`ece (N.S.)}, 27:133--136, 1986.

\bibitem{Tsi85}
B.~S. Tsirelson.
\newblock Quantum analogues of {B}ell's inequalities. {T}he case of two
  spatially divided domains.
\newblock {\em Zap. Nauchn. Sem. Leningrad. Otdel. Mat. Inst. Steklov. (LOMI)},
  142:174--194, 200, 1985.
\newblock Problems of the theory of probability distributions, IX.

\bibitem{Var74}
N.~T. Varopoulos.
\newblock On an inequality of von {N}eumann and an application of the metric
  theory of tensor products to operators theory.
\newblock {\em J. Functional Analysis}, 16:83--100, 1974.

\end{thebibliography}

\end{document}